\documentclass{article}
\def\final{0}
\def\com{1}

\def\alt{0}
\def\cm{0}
 \usepackage[final, nonatbib]{nips_2016}

\usepackage[utf8]{inputenc} % allow utf-8 input
\usepackage[T1]{fontenc}    % use 8-bit T1 fonts
\usepackage{hyperref}       % hyperlinks
\usepackage{url}            % simple URL typesetting
\usepackage{booktabs}       % professional-quality tables
\usepackage{amsfonts}       % blackboard math symbols
\usepackage{nicefrac}       % compact symbols for 1/2, etc.
\usepackage{microtype}      % microtypography

\usepackage{amsmath, amssymb, amsthm}
\usepackage{mathtools}
\usepackage{framed}
\usepackage{enumitem}
\usepackage{verbatim}
\usepackage{color}
\usepackage{caption}
\usepackage{booktabs}
\usepackage{subcaption}
\usepackage{bbm}
\usepackage{mathrsfs}

\ifnum\com=0
\newcommand{\mynote}[1]{\marginpar{\tiny\sf #1}}
\else
\newcommand{\mynote}[1]{}
\fi
\newcommand{\snote}[1]{\mynote{\color{blue}Steven: {#1}}}
\newcommand{\rynote}[1]{\mynote{\color{red}Ryan: {#1}}}
\newcommand{\sj}[1]{\mynote{\color{magenta}Shahin: {#1}}}
\newcommand{\ar}[1]{\mynote{\color{brown}Aaron: {#1}}}
\newcommand{\Prob}[2]{\underset{#1}{\mathbb{P}}\left[ #2 \right]}

\newcommand{\bbe}{\mathbf{e}}

\newcommand{\bp}{\mathbf{p}}

\newcommand{\bbw}{\mathbf{w}}
\newcommand{\bby}{\mathbf{y}}
\newcommand{\bbx}{\mathbf{x}}
\newcommand{\bbz}{\mathbf{z}}
\newcommand{\bbb}{\mathbf{b}}
\newcommand{\bbc}{\mathbf{c}}

\newcommand{\bbv}{\mathbf{v}}

\newcommand{\uttb}{^{(t-1)}}
\newcommand{\utt}{^{(t)}}
\newcommand{\uttf}{^{(t+1)}}

\newcommand{\learnedge}{\texttt{LearnEdge}}
\newcommand{\learnhull}{\texttt{LearnHull}}
\newcommand{\adversary}{\texttt{ADVERSARY}}
\newcommand{\conmatr}{\texttt{MATRIX}}
\newcommand{\learnell}{\texttt{LearnEllipsoid}}
\newcommand{\nac}{\texttt{NAC}}
\newcommand{\ad}{\texttt{AD-2}}

\DeclareMathOperator*{\myargmin}{\mathrm{argmin}}
\DeclareMathOperator*{\myargmax}{\mathrm{argmax}}

\usepackage{algorithm}
\usepackage{algpseudocode}
\algtext*{EndWhile}% Remove "end while" text
\algtext*{EndFor}
\algtext*{EndIf}% Remove "end if" text
\algnewcommand\algorithmicinput{\textbf{Input:}}
 \algnewcommand\INPUT{\item[\algorithmicinput]}
 \algnewcommand\algorithmicoutput{\textbf{Output:}}
 \algnewcommand\OUTPUT{\item[\algorithmicoutput]}

\newcommand\R{\mathbb{R}}

\newcommand{\ellip}{\texttt{ELLIPSOID}}
\newcommand{\cent}{\texttt{CENTROID}}
\newcommand{\cA}{\mathcal{A}}

\newcommand{\cC}{\mathcal{C}}
\newcommand{\cD}{\mathcal{D}}
\newcommand{\cE}{\mathcal{E}}
\newcommand{\cF}{\mathcal{F}}

\newcommand{\cH}{\mathcal{H}}
\newcommand{\cI}{\mathcal{I}}

\newcommand{\cK}{\mathcal{K}}
\newcommand{\cL}{\mathcal{L}}

\newcommand{\cN}{\mathcal{N}}

\newcommand{\cP}{\mathcal{P}}

\newcommand{\tN}{\widetilde{\cN}}
\newcommand{\poly}{\mathrm{poly}}

\newcommand{\argmin}{\arg\!\min}
\newcommand{\argmax}{\arg\!\max}

\newtheorem{theorem}{Theorem}[]
\newtheorem{lemma}[theorem]{Lemma}

\newtheorem{corollary}{Corollary}

\newtheorem{assumpt}{Assumption}
\theoremstyle{definition}
\newtheorem{definition}{Definition}

\title{Learning from Rational Behavior: \\Predicting Solutions to Unknown Linear Programs}

\author{
  Shahin Jabbari, Ryan Rogers, Aaron Roth, Zhiwei Steven Wu\\
  University of Pennsylvania\\
  \texttt{\{jabbari@cis, ryrogers@sas, aaroth@cis, wuzhiwei@cis\}.upenn.edu} \\
}

\begin{document}

\maketitle

\begin{abstract}
We define and study the problem of predicting the solution to a linear
program (LP) given only partial information about its objective and
constraints. This generalizes the problem of learning to predict the
purchasing behavior of a rational agent who has an unknown objective
function, that has been studied under the name ``Learning from
Revealed Preferences". We give mistake bound learning algorithms in
two settings: in the first, the \emph{objective} of the LP
is known to the learner but there is an arbitrary, fixed set of
constraints which are unknown. Each example is
defined by an additional known constraint and the goal of the
learner is to predict the optimal solution of the LP given
the union of the known and unknown constraints. This models 
the problem of predicting the behavior of a rational
agent whose \emph{goals} are known, but whose resources are
unknown. In the second setting, the objective of the LP is
unknown, and changing in a controlled way. The \emph{constraints} of
the LP may also change every day, but are known. An
example is given by a set of constraints and partial information about
the objective, and the task of the learner is again to predict the
optimal solution of the partially known LP.
\end{abstract}

\section{Introduction}
\label{sec:intro}

We initiate the systematic study of a general class of 
multi-dimensional prediction problems,
%\sj{this is not quite a
 % classification problem but rather a prediction one?}\ar{What is the
 % difference?}
 %\sj{I am being too nit-picky, I guess but I think classification is when you have some fixed
 %classes/categories to choose from.}
 %\ar{The $\ell_\infty$ ball isn't good enough? :-p Fine, prediction it is.}
 where the learner wishes to predict the solution to an
unknown linear program (LP), given some partial information
about either the set of constraints or the objective.  In the special
case in which there is a single known \emph{constraint} that is
changing and the objective that is unknown and
fixed, this problem has been studied under the name \emph{learning
  from revealed preferences} \cite{ACDKR15,BDMUV15,BV06,ZR12} and
captures the following scenario: a buyer, with an unknown linear
utility function over $d$ goods $u:\mathbb{R}^d\rightarrow \mathbb{R}$
defined as $u(\bbx) = \bbc\cdot \bbx$ faces a purchasing decision
every day. On day $t$, she observes a set of prices $\bp^t \in
\mathbb{R}^d_{\geq 0}$ and buys the bundle of goods that maximizes
her unknown utility, subject to a \emph{budget} $b$:
$$\bbx^{(t)} = \argmax_{\bbx} \bbc \cdot \bbx\quad\quad\quad\textrm{such that } \bp^t \cdot \bbx \leq b$$
%\sj{can we write $u(x)$ as $u\cdot x$ since $u$ is linear and we are claiming that revealed pref is a special case of ours.}
%% \sj{nit: can we use some other letter for $B$ since we denote matrices by capital letter?}
%$$\textrm{such that } \langle \bbx, \bp^t \rangle \leq B$$
In this problem, the goal of the learner is to predict the bundle that
the buyer will buy, given the prices that she faces. Each example at day $t$
is specified by the vector $\bp^t \in \mathbb{R}^d_{\geq 0}$ (which
fixes the constraint), and the goal is to accurately predict the
purchased bundle $\bbx^{(t)} \in [0,1]^d$ that is the result of
optimizing the unknown linear objective. %% If the learner makes a mistake,
%% she observes the optimal bundle purchased by the buyer as the feedback.
%Then the learner observes the optimal bundle purchased by the buyer as
%feedback to see whether she makes a mistake.

It is also natural to consider the class of problems in which the
goal is to predict the outcome to a LP
broadly e.g. suppose the objective $\bbc\cdot \bbx $ is
known but there is an \emph{unknown} set of constraints $A \bbx \leq
\bbb$. An \emph{instance} is again specified by a changing known
constraint $(\bp^t, b^t)$ and the goal is to predict:
\begin{equation}
\bbx^{(t)} = \argmax_{\bbx} \bbc\cdot \bbx \quad\quad\quad\textrm{such that } A \bbx \leq \bbb\ \ \ \textrm{and } \bp^t\cdot
\bbx \leq b^t.
\label{eq:knownobj_LP}
\end{equation}

  This models the problem of predicting the behavior of an agent whose
  \emph{goals} are known, but whose resource constraints are
  unknown.
%  \footnote{In fact, the first broad use of linear programming
%    was planning and targeting during WWII. A stylized
%    problem that this scenario models is the following: we wish to
%    predict which cities an enemy will bomb each day. We know the
%    relative military value of each target, which gives us his
%    objective function, but we do not know what resources he has at
%    his disposal. The known, changing constraint each day here might
%    represent weather conditions which can make subsets of the targets
%    difficult to attack.}

%\sj{this sentence is not clear. maybe change it to ``Another natural generalization is when the objective ..."}\ar{Changed}
Another natural generalization is the problem in which the objective is unknown,
and may vary in a specified way across examples, and in which there
may also be multiple arbitrary known constraints which vary across
examples. Specifically, suppose that there are $n$ distinct, unknown
linear objective functions $\bbv^1,\ldots,\bbv^n$. An \emph{instance} on day $t$ 
is specified by a subset of the unknown objective functions, $S^t
\subseteq [n]:= \{1,\ldots,n\}$ and a convex feasible region $\cP^t$, and
the goal is to predict:
\begin{equation}
\bbx^{(t)} = \argmax_{\bbx} \sum_{i \in S^t} \bbv^i \cdot \bbx \quad\quad\quad\textrm{such that } \bbx \in \cP^t.
\label{eq:knownconstr_LP}
\end{equation}

When the changing feasible regions $\cP^t$ correspond simply to
varying prices as in the revealed preferences problem, this models a
setting in which at different times, purchasing decisions are made by
different members of an organization, with heterogeneous preferences
--- but are still bound by an organization-wide budget.  The learner's
problem is, given the subset of decision makers and the
prices at day $t$, to predict which bundle they will purchase. This
generalizes some of the preference learning problems recently studied
by Blum et al~\cite{BMM15}. Of course, in this
generality, we may also consider a richer set of changing constraints
which represent things beyond prices and budgets.

In all of the settings we study, the problem can be viewed as the task
of predicting the behavior of a rational decision maker, who always
chooses the action that maximizes her objective function subject to a
set of constraints. Some part of her optimization problem is unknown,
and the goal is to learn, through observing her behavior, that unknown
part of her optimization problem sufficiently so that we may reliably
predict her future actions.

\subsection{Our Results}
We study both variants of the problem (see below)
%% (in which the constraints are
%% unknown, but the objective is known, and in which the objective is
%% unknown, but the constraints are known)
in the strong \emph{mistake
  bound} model of learning \cite{Littlestone87}. In this model, the
learner encounters an arbitrary adversarially chosen sequence of
examples online and must make a prediction for the optimal solution
in each example before seeing future examples. Whenever the
learner's prediction is incorrect, the learner encounters a
\emph{mistake}, and the goal is to prove an
upper bound on the number of mistakes the learner can make, in
the worst case over the sequence of examples. Mistake bound
learnability is stronger than (and implies) PAC learnability~\cite{Valiant84}.

\paragraph{Known Objective and Unknown Constraints}
 We first study this problem under the assumption that there is a
 uniform upper bound on the number of bits of precision used to
 specify the constraint defining each example. In this case, we show
 that there is a learning algorithm with both running time and mistake
 bound linear in the number of edges of the polytope
 \rynote{Also say linear in precision of input?}\sj{agree with
   precision.}formed by the
 unknown constraint matrix $A \bbx \leq b$. We note that this is
 always polynomial in the dimension $d$ when the number of unknown
 constraints is at most $d + O(1)$. 
 (In \ifnum\cm=1 the supplementary material\else Appendix~\ref{sec:poly-mistake}\fi, we show that by allowing the learner to
run in time exponential in $d$, we can give a mistake bound that is always linear in the dimension and
the number of rows of $A$, but we leave as an open question whether or not this mistake bound can be
achieved by an efficient algorithm.)
%\ifnum\final=0 
%\ifnum\alt=0
%(In Appendix~\ref{sec:poly-mistake}, we 
%\else
%(We can
%\fi
% show that by
% allowing the learner to run in time exponential in $d$, we can give a
% mistake bound that is always linear in the dimension and the number
% of rows of $A$, but we leave as an open question whether or not this
% mistake bound can be achieved by an efficient
% algorithm). \else
%(In the full version, we show that by allowing the learner to run in
% time exponential in $d$, we can give a mistake bound that is always
% linear in the dimension and the number of rows of $A$, but 
%With a randomized halving technique one can show that there exists a learner with
%a mistake bound that is always linear in the dimension and the number
%of rows of $A$ but it's running time is exponential in $d$.
%We leave
%as an open question whether or not this mistake bound can be achieved
%by an efficient algorithm. 
  We then show \ifnum\final=1 in the full version \fi that our bounded precision assumption is necessary ---
 i.e. we show that when the precision to which constraints are
 specified need not be uniformly upper bounded, then no algorithm for
 this problem in dimension $d \geq 3$ can have a finite mistake bound.

\ifnum\final=0
 This lower bound motivates us to study a PAC style variant of the
 problem, where the examples are not chosen in an adversarial manner,
 but instead are drawn independently at random from an arbitrary
 unknown distribution. In this setting, we show that even if the
 constraints can be specified to arbitrary (even infinite) precision,
 there is a learner that requires sample complexity only
 linear in the number of edges of the unknown constraint polytope. This learner
 can be implemented efficiently when the constraints are specified with finite precision.

\fi
\paragraph{Known Constraints and Unknown Objective}
For the variant of the problem in which the objective is unknown and
changing and the constraints are known but changing, we give an
algorithm that has a mistake bound and running time 
polynomial in the dimension $d$. Our algorithm uses the Ellipsoid
algorithm to learn the coefficients of the unknown
objective by implementing a separation oracle that generates
separating hyperplanes given examples on which our algorithm made
a mistake.

We leave the study of either of our problems under natural relaxations (e.g. under a less demanding loss function)
and whether it is possible to substantially improve our results in these relaxations
as an interesting open problem.

\subsection{Related Work}
Beigman and Vohra~\cite{BV06} were the first to study \emph{revealed
preference problems} (RPP) as a learning problems and to relate them to
multi-dimensional classification. They %% ~\cite{BV06}
derived sample complexity bounds for such problems
by computing the fat shattering dimension of the class of
target utility functions, and showed that the set of
Lipschitz-continuous valuation functions had finite fat-shattering
dimension. 
Zadimoghaddam and Roth~\cite{ZR12} gave efficient
algorithms with polynomial sample complexity for PAC learning of the
RPP over the class of linear (and piecewise
linear) utility functions. Balcan et al.~\cite{BDMUV15} showed a
connection between RPP and the
structured prediction problem of learning $d$-dimensional linear
classes \cite{Col00,Col02,LMP01}, and use an efficient variant of the
compression techniques given 
by Daniely and Shalev-Shwartz \cite{DS14}
to give efficient PAC algorithms with optimal sample complexity for
various classes of economically meaningful utility functions. 
Amin et al.~\cite{ACDKR15} 
study the RPP for linear
valuation functions in the mistake bound model, and in the
query model in which the learner gets to set prices and wishes to
maximize profit. 
Roth et al.~\cite{RUW15} also study the query model
of learning and give results for strongly concave objective
functions, leveraging an algorithm of 
Belloni et al.~\cite{BLNR15} for bandit convex optimization with adversarial
noise.

All of the works above focus on the setting of predicting the
optimizer of a fixed unknown objective function, together with a
single known, changing constraint representing prices. This is the
primary point of departure for our work --- we give algorithms for the
more general settings of predicting the optimizer of a LP
when there may be many unknown constraints, or when the unknown
objective function is changing. Finally, the literature on
\emph{preference learning} (see e.g. \cite{preflearning}) has similar
goals, but is technically quite distinct: the canonical problem in
preference learning is to learn a \emph{ranking} on distinct
elements. In contrast, the problem we consider here is to predict the
outcome of a continuous optimization problem as a function of varying
constraints.  \snote{think for 3 authors, we should also say et al.}

\section{Model and Preliminaries}
\label{sec:model}
%\subsection{Preliminaries}
%\label{sec:prelim}
%For the rest of the paper, we will make use of the following elementary geometric concepts. A \emph{hyperplane} and a \emph{halfspace} in $\R^d$ are the set of points satisfying some linear equation $a_1 x_1 + \ldots a_d x_d = b$ and some linear inequality $a_1 x_1 + \ldots + a_dx_d \leq b$ respectively, assuming that not all $a_i$'s are simultaneously zero. A \emph{(convex)  polyhedron} is an intersection of finitely many halfspaces.   A bounded polyhedron is called a \emph{polytope}.  The boundary of a polyhedron consists of convex polyhedra of lower dimension called \emph{faces}.  A $k$-face denotes a $k$-dimensional face. The $(d-1)$-faces of an $d$-dimensional polyhedron are called \emph{facets},  the $1$-faces and $0$-faces are called \emph{edges} and \emph{vertices} respectively.

We first formally define the geometric notions used throughout this paper.~A \emph{hyperplane} and a \emph{halfspace} in $\R^d$ are the set of
points satisfying the linear equation $a_1 x_1 + \ldots a_d x_d = b$
and the linear inequality $a_1 x_1 + \ldots + a_dx_d \leq b$ for a set of $a_i$s
respectively, assuming that not all $a_i$'s are simultaneously zero.
%% We call the a set of points that satisfy a linear inequality $a_1
%% x_1 + \ldots + a_dx_d \leq b$ as a \emph{halfspace} and the set of
%% points that make the inequality tight a \emph{hyperplane}.
A set of \emph{hyperplanes} are \emph{linearly independent} if the normal vectors to the hyperplanes
are linearly independent.
%if their
% coefficients vectors $(a_1, \ldots, a_d)$ are linearly
%independent.
%\ar{Should be more precise about what this vector is. Is
  %it the vector of coefficients $(a_1,\ldots,a_d)$, or does it also
  %include $b$?}
  %\snote{reworked}\rynote{Used gradient} 
A \emph{polytope} (denoted by $\cP \subseteq \R^d$) is the \emph{bounded} intersection of finitely many
halfspaces, written as $\cP = \{\bbx\mid A\bbx\leq \bbb\}$.  An
\emph{edge-space} $e$ of a polytope $\cP$ is a one dimensional
subspace that is the intersection of $d-1$ linearly independent
hyperplanes of $\cP$, and an \emph{edge} is the intersection between an
edge-space $e$ and the polytope $\cP$.\iffalse\footnote{Note that an \emph{edge} on
  the boundary of the polytope $\cP$ is a line segment contained in an
  edge-space of $\cP$.}\fi   
  We denote the set of edges of polytope $\cP$ by $E_\cP$.
%\ar{Confused by the foot note. Is this not the standard definition of an edge? If not, should we use a different
  %word?}
%\rynote{Use edge-space.}\snote{also define edge}
A \emph{vertex} of $\cP$ is a point where
$d$ linearly independent hyperplanes of $\cP$ intersect.  
Equivalently, $\cP$ can be written as the \emph{convex hull}
of its vertices $V$ denoted by $\mathrm{Conv}(V)$.  Finally, we define a set of points to be
\emph{collinear} if there exists a line that contains all the points
in the set.

We study an online prediction problem with the goal of 
predicting the optimal solution of a changing LP whose
parameters are only partially known.  Formally, in each day
$t=1,2, \ldots$ an adversary chooses a LP specified by a polytope $\cP\utt$ % = \{\bbc\in \R^d: A\utt\bbx\leq\bbb\utt \}$
(a set of linear inequalities) and coefficients $\bbc\utt \in \R^{d}$ of the linear
objective function.  The learner's goal is to predict the solution
$\bbx\utt$ where
%\ifnum\final=0
%\begin{align}
%\bbx\utt = \argmax_{\bbx \in \cP\utt } \quad \bbc\utt\cdot\bbx.\label{eq:LP}
%\end{align}
%\else
$\bbx\utt = \argmax_{\bbx \in \cP\utt } \bbc\utt\cdot\bbx.$
%\fi
%\rynote{reworked.  The number of days we run does not seem to matter in the mistake bound.}
After making the prediction $\hat \bbx\utt$, the learner observes the
optimal $\bbx\utt$ and learns whether she has made a mistake
$(\hat\bbx\utt\neq \bbx\utt)$.
The mistake bound is defined as follows.
\sj{we never use the definition of $\sigma$ anywhere else in the paper. So either state at the beginning of each section
what $\sigma$ is or drop this notation. I vote for the latter!}
\begin{definition}
Given a LP with feasible polytope $\cP$ and objective function
$\bbc$, let $\sigma\utt$ denote the parameters of the LP that
are revealed to the learner on day $t$.  
A learning algorithm $\cA$ takes as input the sequence $\{\sigma\utt\}_t$, the known
parameters of an adaptively chosen sequence $\{(\cP\utt,\bbc\utt) \}_t$ of LPs and
outputs a sequence of predictions $\{\hat\bbx\utt\}_t$.  We say that
$\cA$ has mistake bound $M$ if
 $\max_{ \{(\cP\utt,
  \bbc\utt)\}_t } \left\{ \Sigma_{t=1}^\infty \mathbf{1}\left[\hat \bbx\utt \neq
  \bbx\utt\right] \right\}\leq M,
$
\ifnum\final=0 where $\bbx\utt = \argmax_{\bbx \in \cP\utt } \bbc\utt\cdot\bbx$ on day $t$.\fi
\end{definition}

%  \rynote{Reworked.  I think this is all we need to say about the full generality setting (i.e. just give equation 3 and the mistake bound definition).  We then say how it relates to our setting.}
\ifnum\final=0
We consider two different instances of the problem described above.
First, in Section \ref{sec:known_obj}, we study the problem given in
\eqref{eq:knownobj_LP} in which $\bbc\utt = \bbc$ is fixed and known
to the learner but the polytope $\cP\utt = \cP \cap \cN\utt$ consists
of an unknown fixed polytope $\cP$ and a new constraint $\cN\utt =
\{\bbx \mid \bp\utt\cdot\bbx\leq b\utt \}$ which is revealed to the
learner on day $t$ i.e. $\sigma\utt = (\cN\utt , \bbc)$. We refer to this as the \emph{Known Objective}
problem.
Then, in Section \ref{sec:known_constr}, we study the
problem in which the polytope $\cP\utt$ is changing and known but the
objective function $\bbc\utt = \sum_{i \in S\utt} \bbv^i $ is unknown and changing as in
\eqref{eq:knownconstr_LP} where the set $S\utt$ is known i.e.  $\sigma\utt = (\cP\utt , S\utt)$. We refer to this as the \emph{Known
  Constraints} problem. 
\fi

\ifnum\final=0
In order for our prediction problem to be well defined, we make 
Assumption~\ref{assumpt:vertex} about the observed solution $\bbx\utt$ in each day.
\else
Throughout the paper, we assume the uniqueness of $\bbx\utt$ in order for the prediction problem to be well defined.
\fi
\ifnum\final=0
Assumption~\ref{assumpt:vertex} guarantees that each solution is on a
vertex of $\cP\utt$. %\rynote{Moved precision assumption.}
\fi
\begin{assumpt}
The optimal solution to the LP: $\max_{\bbx\in \cP\utt} \bbc\utt\cdot\bbx$ is unique for all $t$.% and is on a vertex of the polytope $\cP\utt$. \ar{There is always a vertex optimal solution, right? This doesn't seem like an assumption -- the only assumption is that the solution is unique.}
\label{assumpt:vertex}
\end{assumpt}
%\sj{changed $[\bbc\cdot\bbx]$ to $\bbc\utt\cdot\bbx$ in the definition}

\section{The Known Objective Problem}
\label{sec:known_obj}

In this section, we focus on the \emph{Known Objective Problem} where
the coefficients of the objective function $\bbc$ are fixed and known
to the learner but the feasible region $\cP\utt$ 
on day $t$ is unknown and changing. In particular, $\cP\utt$
is the intersection of a fixed and unknown polytope
$\cP = \{\bbx\mid A\bbx\leq \bbb, A\subseteq \R^{m\times d}\}$ and a known halfspace $\cN\utt =
\{\bbx \mid \bp\utt \cdot \bbx\leq b\utt\}$
% specified by a changing
%constraint that is known to the learner 
i.e. $\cP\utt =
\cP\cap \cN\utt$. 
%% linear constraint that is specified to the learner
%% with each example.
%\snote{should we give some motivation for known
%  objective? as we do on the known constraint problem section}

 %% More specifically, the fixed polytope $\cP$ can be written as a
 %% collection of linear inequalities which are unknown to the learner ,
 %% that is $\cP=\{\bbx \mid A\bbx\le \bbb\}$.  

Throughout this section we make the following assumptions. First, we 
assume w.l.o.g. (up to scaling) that
the points in $\cP$ have $\ell_\infty$-norm
bounded by 1.  
%\rynote{Changed norm to $\ell\infty$.  This will change later bounds.}
\begin{assumpt}
The unknown polytope $\cP$ lies inside the unit $\ell_\infty$-ball i.e. $\cP \subseteq
\{\bbx \mid ||\bbx||_\infty \leq 1 \}$.
\label{assumpt:poly}
\end{assumpt}
\sj{do we actually use the name of this assumptions other than solution uniqueness? If not, we probably should remove the names!}
%\rynote{Moved assumption here from prelim section to only deal with $\cP$ vertices.}
We also assume that the
coordinates of the vertices in $\cP$ can be written
with finite precision (this is implied if the halfspaces defining $\cP$ can be described with finite precision). \footnote{Lemma 6.2.4
  from Grotschel et al.~\cite{GLS93} states that if each constraint in $\cP\subseteq\R^d$ has
  encoding length at most $N$ then each vertex of $\cP$ has
  encoding length at most $4d^2N$.  Typically 
  the  finite
  precision assumption is made on the constraints of the LP. However, since this assumption implies 
  that the vertices can be described with finite
  precision, for simplicity, we make our assumption directly on the
  vertices.}
  \sj{I wrote this footnote differently and now it is changed. I would change ``However, since this assumption ... " to ``However, Lemma 6.2.4 from~\cite{GLS93}
  has a converse. Since it is more convenient for us to assume precision on the vertices of the polytope, we use this assumption throughout the paper without loss
  of generality." Thoughts?}

\begin{assumpt}
The coordinates of each vertex of $\cP$ can be written with $N$ bits of precision.
\label{assumpt:precision}
\end{assumpt}
%\snote{removed this paragraph from NIPS}
\ifnum\final=0
We show in Section \ref{sec:lowerbound} that Assumption
\ref{assumpt:precision} is necessary --- without any upper
bound on precision, there is no algorithm with a finite mistake bound.
\fi
Next, we make some non-degeneracy assumptions on
polytopes $\cP$ and $\cP\utt$, respectively. We require these
assumptions to hold on each day.
%\footnote{Assumptions~\ref{assumpt:rank}~and~\ref{assumpt:bind}
%are standard assumptions}
%% Our first assumption is on the
%% constraint matrix $A$ defining the fixed unknown polytope $\cP$:
\begin{assumpt}
Any subset of $d-1$ rows of $A$ have rank $d-1$
where $A$ is the constraint matrix in $\cP = \{\bbx \mid A\bbx \leq \bbb \}$.
\label{assumpt:rank}
\end{assumpt}
\begin{assumpt}
Each vertex of $\cP\utt$ is the intersection of exactly
$d$-hyperplanes of $\cP\utt$.
\label{assumpt:bind}
\end{assumpt}
%\snote{made some cut here}
%We now formally define the problem we consider in this section.
%\rynote{May want to modify this definition but just wanted all the assumptions we use to be stated upfront so that
%we can just refer to this definition for each Lemma.}
%\begin{definition}[The Known Objective Problem]
%Subject to the solution uniqueness and bounded precision assumptions \ref{assumpt:vertex} and \ref{assumpt:precision}, and the non-degeneracy assumptions \ref{assumpt:rank} and \ref{assumpt:bind}, an adversary chooses a fixed polytope $\cP = \{\bbx : A\bbx\leq \bbb \}$ and an objective $\bbc$, and at each day $t$ chooses a constraint $\cN\utt = \{ \bbx : \bp\utt \cdot \bbx \leq b\utt\}$. The objective $\bbc$ is known to the learner, and the constraint $\cN\utt$ is revealed to the learner at each day, but the polytope $\cP$ is unknown to the learner. At each day after learning $\cN\utt$, the learner's goal is to predict $\bbx\utt$, the solution to $\max_{\bbx \in \cP\utt} \bbc \cdot \bbx$
%\end{definition}
%\sj{same issue with def of mistake bound as before.}

\ifnum\final=0
The rest of this section is organized as follows. 
We present
$\learnedge$ for the Known Objective Problem
and analyze its mistake bound in Sections~\ref{sec:learn-edge}~and~\ref{sec:edge_analysis}, respectively.   Then in
Section~\ref{sec:lowerbound}, we prove the necessity of Assumption \ref{assumpt:precision} to get a finite mistake bound. 
Finally in Section
\ref{sec:stochastic}, we present the $\learnhull$ in a
PAC style setting where the new constraint each day is drawn
i.i.d. from an unknown distribution, rather than selected
adversarially.  %We show that $\learnhull$ has sample complexity linear
%in the number of edges of $\cP$ even without Assumption \ref{assumpt:precision}.  
 \fi

\subsection{$\learnedge$ Algorithm}
\label{sec:learn-edge}
In this section we introduce $\learnedge$ and  show in Theorem~\ref{thm:learn-edge}
that the number of mistakes of $\learnedge$ depends linearly on the number of edges $E_\cP$ and the
precision parameter $N$ and only logarithmically on the dimension $d$.
%We introduce $\learnedge$ and 
%prove a mistake bound for it  that is linear on $|E_\cP|$ and $N$ and logarithmic on $d$.
We defer \ifnum\cm=1 all the missing proofs to the supplementary material.
\else all the proofs of this section to Appendix~\ref{sec:missingproof3}.\fi
%\rynote{Introduced notation of $E_\cP$.}

%\rynote{Do you think we need to say the Known Objective Problem in
 % this theorem?  I guess it is obvious but just want the theorem to be
 % self contained.}
\begin{theorem}
  The number of mistakes and per day running time of $\learnedge$ in the Known Objective Problem are
$O(|E_\cP| N\log(d))$ and $\poly(m,d, |E_\cP|)$ respectively when $A\subseteq \R^{m\times d}.$
  \label{thm:learn-edge}
\end{theorem}
\sj{what is $m$. We have never defined this before!}
\rynote{Good point!  Must have been copy and pasted from a section where $m$ was defined.  Added earlier discussion}

%\ar{Does the running time actually scale with $T$, or just the number
 %of edges learned so far? The current statement makes it seem like
 %the running time can increase beyond $|E|$.}
%\snote{put every day here for ``runtime''. }
%\rynote{Ensuring that updates are only done when we make a mistake. }
%\snote{I think it runtime should be poly in E as well; updated the theorem statement}

\ifnum\final=1
At a high level, $\learnedge$ makes a prediction and then updates its information in each day.
The three main components of $\learnedge$ are as follows.

\begin{enumerate}
\item\label{first} \textbf{Prediction Information}:
 $\learnedge$ maintains information $\cI\utt$ about the prediction history up to day $t$. 
  The primary object in $\cI\utt$ is the set of edges-spaces of 
  $\cP$ since the observed solutions 
  always lie on edges of $E_\cP$ (see
  Lemma~\ref{lem:edge}). 
\iffalse
Furthermore, the set of edge-spaces are easy to
  learn --- as we show in Lemma~\ref{lem:collinear}, we can identify
  an edge-space of $\cP$ whenever we see three collinear observed solutions.
\fi
\item \textbf{Prediction Rules}: Given the additional halfspace
  $\cN\utt$ and the information $\cI\utt$, $\learnedge$
  makes its prediction $\hat\bbx\utt$ based on
  \emph{prediction rules}  \ref{p1} - \ref{p4}.% (Which we will specify later).
  %which can be broken down into four cases as
  %we show.
 % \rynote{Update only when we have mistake.}
\item \textbf{Update Rules}: 
$\learnedge$ will update the
  information to $\cI\uttf$ with its \emph{updating rule}
  upon making a mistake on day $t$ using the 
  observed optimal point $\bbx\utt$. 
  %$\learnedge$ will use this updated information to make a prediction the next day. %There are also 4
  %cases, and
  Our mistake bound crucially relies on the fact that each
  update will lead to progress on learning the feasible region of the
  edge-spaces (i.e. the edges $E_\cP$ of $\cP$).
  \end{enumerate}
\else

 At a high level, $\learnedge$ maintains a set of \emph{prediction
   information} $\cI\utt$ about the prediction history up to day
 $t$, and makes prediction in each day based on $\cI\utt$ and a set
 of \emph{prediction rules} $(\ref{p1}-\ref{p4})$. After making a mistake,
 $\learnedge$ updates the information with a set of \emph{update rules}
 $(\ref{u1}-\ref{u4})$.
 \fi
\ifnum\cm=0
%\rynote{Feel free to take out this formal description of the algorithm if it is already clear what the algorithm is, I just thought I would put it in.}\snote{TODO: update only happens when mistake occurs; not in NIPS version}
The framework of $\learnedge$ is presented in Algorithm~\ref{alg:learnedge}. We will now present the details of each component.
\begin{algorithm}[h]
 \begin{algorithmic}
 \Procedure{$\learnedge$}{$\{\cN\utt ,\bbx\utt\}_t$}\ifnum\alt=0\hspace{41mm}\else\hspace{15mm}\fi$\triangleright$ Against adaptive adversary
	\State Initialize $\cI^{(1)} $ to be empty.	\ifnum\alt=0\hspace{76mm}\else\hspace{50mm}\fi$\triangleright$ Initialize
	\For{$t = 1,2, \cdots$}\\
	\hspace{10mm}Predict $\hat\bbx\utt$ according to one of \ref{p1} - \ref{p4} based on $\cN\utt$ and $\cI\utt$.
		\ifnum\alt=0\hspace{20mm}\else \hspace{94mm}\fi$\triangleright$ Predict
		\If{$\hat\bbx\utt \neq \bbx\utt$}\\
			\hspace{20mm}$X\uttf \gets X\utt \cup \{\bbx\utt \}.$\\
			\hspace{20mm}Update $\cI\uttf$ with \ref{u1} - \ref{u4} based on $\bbx\utt$.
				\ifnum\alt=0\hspace{40mm}\else\hspace{16mm}\fi$\triangleright$ Update
		\EndIf
	\EndFor
 \EndProcedure
\end{algorithmic}
\caption{Learning in Known Objective Problem ($\learnedge$)}
\label{alg:learnedge}
\end{algorithm}
\fi

\paragraph{Prediction Information}
It is natural to ask ``What information
is useful for prediction?" Lemma~\ref{lem:edge} establishes the
importance of the set of edges $E_\cP$ by showing that all the
observed solutions will be on an element of $E_\cP$.
%\ifnum\alt=0
%We defer the omitted proofs of this section to Appendix~\ref{sec:missingproof3}.
%\else
%We defer all the proofs to the supplementary material.
%\fi
\begin{lemma}
\label{lem:edge}
On any day $t$, the observed solution $\bbx\utt$ lies on an edge in $E_\cP$.
\end{lemma}
%\ifnum\final=0
%\input{./lem-edge}
%\fi

%\rynote{added.}
In the proof of Lemma~\ref{lem:edge} we also show that when $\bbx\utt$
does not bind the new constraint $\cN\utt$, then $\bbx\utt$ is the solution for
the underlying LP: $\argmax_{\bbx\in \cP} \bbc \cdot \bbx$.
\begin{corollary}
If $\bbx\utt  \in \{\bbx \mid \bp\utt \bbx < b\utt \}$ then $\bbx\utt = \bbx^* \equiv \myargmax_{\bbx \in \cP} \bbc \cdot \bbx$.
\label{cor:opt_t}
\end{corollary}

We then show how an 
edge-space $e$ of $\cP$ can be recovered after seeing 3 collinear observed solutions. 
%\snote{leave the lemma here for now, and remove the proof in NIPS version}
\begin{lemma}
  Let $\bbx, \bby,\bbz$ be 3 distinct collinear points on edges of $\cP$. Then
  they are all on the \emph{same} edge of $\cP$ and the
  1-dimensional subspace containing them is an edge-space of
  $\cP$.
  \label{lem:collinear}
\end{lemma}
%moved the proof to appendix

Given the relation between observed solutions and edges, the
information $\cI\utt$ is stored as follows:
    \begin{figure}[h]
   \centering
   \begin{minipage}[c]{.4\textwidth}\scalebox{.75}{
      \includegraphics[width=6cm]{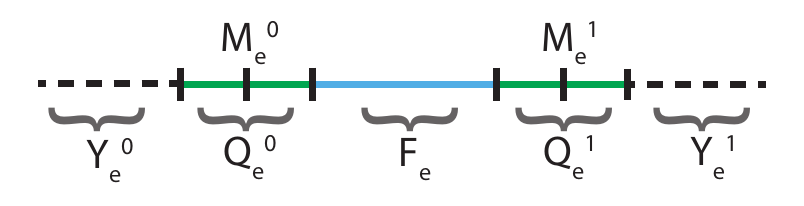}
      }
     \end{minipage}
\begin{minipage}[c]{.5\textwidth} 
    \caption{Regions on an edge-space $e$: feasible
      region $F_e$ (blue), questionable intervals $Q_e^0$ and
      $Q_e^1$ (green) with their mid-points $M_e^0$ and $M_e^1$ 
      and infeasible regions $Y_e^0$ and $Y_e^1$ (dashed).}
    \label{fig:illus}
    \end{minipage}
  \end{figure}

\begin{enumerate}[label=\textbf{I.\arabic*}]

\item({\bf Observed Solutions}) $\learnedge$ keeps track of the set of
  observed solutions that were predicted incorrectly so far $X\utt = \{\bbx^{(\tau)} : \tau \leq t \quad \hat\bbx^{(\tau)} \neq \bbx^{(\tau)}\}$ 
  and also the solution for the underlying unknown
  polytope $\bbx^* \equiv \argmax_{\bbx\in\cP} \bbc\cdot\bbx$ if
  it is observed.

\item({\bf Edges}) $\learnedge$ keeps track of the set of edge-spaces $E\utt$ given by any 3
  collinear points in $X\utt$. For each $ e\in E\utt$, it also
  maintains the regions on $e$ that are certainly \emph{feasible} or \emph{infeasible}. 
  The remaining parts of $e$ called the \emph{questionable} region is where $\learnedge$ cannot classify as 
  infeasible or feasible with certainty (see Figure~\ref{fig:illus}). Formally,
  \end{enumerate}
  
%\footnote{ For the
%  edge-space information, we add the superscript $(t)$ to show the
%  dependence of these quantities on days, and eliminate the
%  subscript $e$ when we take a union over all elements in $E\utt$,
%  e.g. $F\utt = \bigcup_{e \in E\utt} F_e\utt$. }
\begin{enumerate}
\item({\bf Feasible Interval}) The \emph{feasible interval} $F_e$ is an
  interval along $e$ that is identified to be on the
  boundary of $\cP$. More formally, $F_e =\mathrm{Conv}(X\utt \cap e)$.

\item({\bf Infeasible Region}) The \emph{infeasible region}
  $Y_e = Y_e^0 \cup Y_e^1$ is the union of two disjoint intervals
  $Y_e^0$ and $Y_e^1$ that are identified to be outside of
  $\cP$. By Assumption~\ref{assumpt:poly}, we initialize the
  infeasible region $Y_e$ to $\{x\in e\mid
  \|x\|_\infty >1\}$ for all $e$.

\item({\bf Questionable Region}) The \emph{questionable region}
  $Q_e = Q_e^0 \cup Q_e^1$ on $e$ is the union of two disjoint
  \emph{questionable intervals} along $e$. Formally, $Q_e = e \setminus (F_e \cup
  Y_e)$.  The points in $Q_e$ cannot be
  certified to be either inside or outside of $\cP$ by $\learnedge$.

\item({\bf Midpoints in $Q_e$}) For each questionable interval
  $Q_e^i$, let $M_e^i$ denote the midpoint of $Q_e^i$.

\end{enumerate}
\ifnum\final=0 
We add the superscript $(t)$ to show the dependence of
these quantities on days. Furthermore, we eliminate the subscript
$e$ when taking the union over all elements in $E\utt$, e.g. $F\utt =
\bigcup_{e \in E\utt} F_e\utt$.  So the information
$\cI\utt$ can be written as follows:
$
\cI\utt = \left(X\utt, E\utt, F\utt, Y\utt, Q\utt, M\utt \right).
\label{eq:info}
$\fi

\paragraph{Prediction Rules}
We now focus on the prediction rules of $\learnedge$. On day $t$, let ${\tN\utt}
= \{\bbx\mid \bp\utt\cdot\bbx= b\utt\}$ be the hyperplane specified by
the additional constraint $\cN\utt$. If $\bbx\utt\notin\tN\utt$, then
$\bbx\utt=\bbx^*$ by Corollary~\ref{cor:opt_t}.
So whenever the algorithm observes $\bbx^*$, it will
store $\bbx^*$ and predict it in the future days when $\bbx^*\in \cN\utt$. This is case~\ref{p1}.
So in the remaining cases we know $\bbx^*\notin \cN\utt$. 

The analysis of Lemma~\ref{lem:edge}  shows that $\bbx\utt$ must
be in the intersection between $\tN\utt$ and the edges
$E_\cP$, so $\bbx\utt = \argmax_{\bbx\in \tN\utt\cap
E_\cP} \bbc \cdot \bbx$. Hence, $\learnedge$
can restrict its prediction to the following \emph{candidate} set:
%% \rynote{Hmmm... you say these are intersection points, but these are lines and points... Changing this candidate solutions to be intersected with the hyperplane $\tN\utt$.  That way it is a set of points. }
$\mathrm{Cand}\utt =\{ (E\utt\cup X\utt)\setminus \bar
E\utt \} \cap \tN\utt$ where $\bar E\utt = \{e \in E\utt \mid
e\subseteq \tN\utt\}.$  
As we show in Lemma~\ref{lem:corner-case},
$\bbx\utt$ will not be in $\bar E\utt$,
so it is safe to remove $\bar E\utt$ from $\mathrm{Cand}\utt$.
  \begin{lemma}
  \label{lem:corner-case}
Let $e $ be an edge-space of $\cP$ such that $e\subseteq \tN\utt$, then $\bbx\utt\not\in e$.
\end{lemma}
However, $\mathrm{Cand}\utt$ can be empty or
only contain points in the infeasible regions of the edge-spaces. If so, then there is simply not enough
information to predict a feasible point in $\cP$. Hence, $\learnedge$ 
predicts an arbitrary point outside of $\mathrm{Cand}\utt$. This is case~\ref{p2}.

Otherwise $\mathrm{Cand}\utt$ contains points from the feasible and questionable regions of the edge-spaces. $\learnedge$ predicts from 
a subset of $\mathrm{Cand}\utt$ called the
\emph{extended feasible region} $\mathrm{Ext}\utt$
instead of directly
predicting from $\mathrm{Cand}\utt$.
$\mathrm{Ext}\utt$ contains the whole feasible region and only parts 
of the questionable region on all the edge-spaces in $E\utt\setminus \bar E\utt$.
We will show later that 
this guarantees $\learnedge$ makes progress in
learning the true feasible region on some edge-space upon making a mistake.
More formally, $\mathrm{Ext}\utt$ is the intersection of $\tN\utt$ with the union of 
intervals between the two mid-points $(M_e^0)\utt$ and
$(M_e^1)\utt$ on every edge-space $e\in E\utt\setminus \bar
E\utt$ and all points in
$X\utt$:
$
\mathrm{Ext}\utt = \left\{X\utt \cup \left\{\cup_{e\in E\utt \setminus \bar
E\utt} \mathrm{Conv}\left((M_e^0)\utt, (M_e^1)\utt \right)\right\}\right\} \cap \tN\utt.
$

In \ref{p3}, if $\mathrm{Ext}\utt\ne\emptyset$ then $\learnedge$
predicts the point with the highest objective value in $\mathrm{Ext}\utt$.

Finally, if $\mathrm{Ext}\utt=\emptyset$, then we know $\tN\utt$ only intersects
within the questionable regions of the learned edge-spaces. In this case, $\learnedge$
predicts the intersection point with the lowest objective value, which
corresponds to \ref{p4}. Although it might seem counter-intuitive to predict
the point with the lowest objective value, this guarantees that $\learnedge$ makes progress in
learning the true feasible region on some edge-space upon making a mistake. 
The prediction rules are summarized as follows: %\rynote{Took the parentheses out}

\begin{enumerate}[label=\textbf{P.\arabic*}]
\item\label{p1} First, if $\bbx^*$ is observed and $\bbx^*\in\cN\utt$, then predict $\hat\bbx\utt \gets \bbx^*$;
%\rynote{Changed P.2 back. }
\rynote{Still type mismatch between the set Cand (elements are points) and Y (elements are all points on a semi-infinite line).  So an element of Cand cannot be an element of Y...  How about if $\forall \bbx \in \text{Cand}\utt$ there is an  $e \in E\utt$ s.t. $\bbx\in Y_e\utt$ or Cand$\utt$ is empty then predict...}
%\snote{why do you force yourself to make mistake?? the reviewer would want to know.}
\item\label{p2} Else if $\mathrm{Cand}=\emptyset$ or $\mathrm{Cand}\utt \subseteq \bigcup_{e\in E\utt} Y_e\utt$, then predict any point outside $\mathrm{Cand}\utt$;

\item\label{p3} Else if $\mathrm{Ext}\utt \neq\emptyset$, then predict $\hat \bbx\utt = \myargmax_{\bbx\in \mathrm{Ext}\utt} \bbc \cdot \bbx$;

\item\label{p4} Else, predict 
$\hat \bbx\utt = \myargmin_{\bbx\in \mathrm{Cand}\utt}\bbc \cdot \bbx$.

%\item\label{p2} Else if $\bar X\utt \cup B\utt = \emptyset$, then predict $\hat\bbx\utt$ as an
%  arbitrary point outside of $E\utt$ and $X\utt$;
%
%\item \label{p3}  Else if $\bar X\utt \cup
%  C\utt \neq \emptyset$, then predict $\hat \bbx\utt = \argmax_{\bbx\in (\bar X\utt\cup C\utt)} \bbc \cdot \bbx$;
%
%\item\label{p4} Else\iffalse If $\hat X\utt\cup C\utt = \emptyset$\fi, predict
%$\hat \bbx\utt = \argmin_{x\in B\utt}\bbc \cdot \bbx$.
\end{enumerate}

\paragraph{Update Rules}
Next we describe how $\learnedge$ updates its information.
Upon making a mistake, $\learnedge$  adds $\bbx\utt$
to the set of previously observed solutions $X\utt$ i.e. $X\uttf \gets X\utt \cup \{ \bbx\utt\}$. Then it performs
one of the following four mutually exclusive update rules (\ref{u1}-\ref{u4}) in order.
%The first
%routine update \ref{u0} is simply appending the point $\bbx\utt$ to
%the set of previously observed solutions $X\utt$.
%\snote{I think we should only update after making mistakes; this is cleaner}
%\rynote{Agree.}\snote{changed}
%% The algorithm $\learnedge$ then updates the information using one of
%% four possible updates.
\begin{enumerate}[label=\textbf{U.\arabic*}]

\item\label{u1}
If $\bbx\utt\notin\tN\utt$,
then $\learnedge$ records $\bbx\utt$ as the unconstrained optimal
solution $\bbx^*$.
%\ifnum\alt=1 \label{u1}\fi
\item\label{u2}
Then if $\bbx\utt$ is not on any learned edge-space in
$E\utt$, $\learnedge$ will try to learn a new edge-space 
by checking the collinearity of $\bbx\utt$ and any couple of points in $X\utt$.
So after this update $\learnedge$ might recover a new edge-space of the polytope.

\end{enumerate}
 If the previous updates were not invoked, then $\bbx\utt$ was on some  learned edge-space $e$. 
$\learnedge$ then compares the objective
values of $\hat{\bbx}\utt$ and $\bbx\utt$ (we know $\bbc\cdot\hat\bbx\utt \ne \bbc \cdot \bbx\utt$ by Assumption~\ref{assumpt:vertex}):
\begin{enumerate}[label=\textbf{U.\arabic*}]\addtocounter{enumi}{2}
\item\label{u3}
If 
$\bbc\cdot \hat \bbx\utt > \bbc\cdot\bbx\utt$, then $\hat\bbx\utt$ must be infeasible and $\learnedge$ then
updates the questionable and infeasible regions for $e$.  
%\ifnum\alt=1 \label{u3}\fi
\item\label{u4}
If
$\bbc\cdot\hat\bbx\utt < \bbc \cdot \bbx\utt$ then
$\bbx\utt$ was outside of
the extended feasible region of $e$.
$\learnedge$ then updates the questionable region and feasible interval on
$e$.
\end{enumerate}
%\ifnum\alt=1 \label{u4}\fi
\ifnum\final=1
\footnote{In the full version, we give a formal argument on why such edge-space
always exist when \ref{u3} or \ref{u4} is invoked.}
\fi

In both of~\ref{u3} and~\ref{u4}, $\learnedge$ will shrink some questionable interval substantially
till the interval has length less than $2^{-N}$ in which case 
Assumption~\ref{assumpt:precision} implies that the interval contains no points.
So $\learnedge$ can update the adjacent feasible region and infeasible interval accordingly.

\subsection{Analysis of $\learnedge$}\label{sec:edge_analysis}
Whenever $\learnedge$ makes a mistake, one of the update rules \ref{u1} - \ref{u4} is
invoked.  So the number of mistakes of $\learnedge$
is bounded by the number of times each update rule is
invoked. %% In the following, we write $E_\cP$ to denote the set of edges for
%% the unknown polytope $\cP$.
The mistake bound
of $\learnedge$ in Theorem~\ref{thm:learn-edge} is hence the sum of mistakes bounds in Lemmas~\ref{lem:u1}-\ref{lem:u34}.

\begin{lemma}%[Mistake Bound of \ref{u1}]
  Update \ref{u1} is invoked at most 1 time.
    \label{lem:u1}
\end{lemma}

\begin{lemma}%[Mistake Bound of \ref{u2}]
  Update \ref{u2} is invoked at most $3|E_\cP|$ times. \footnote{The dependency on $|E_\cP|$ can be improved by replacing it with the set of edges of $\cP$ 
  on which an optimal solution is observed.  This applies to all the dependencies on $|E_\cP|$ in our bounds.}% with optimal points on them.
    \label{lem:u2}
\end{lemma}

\begin{lemma}%[Mistake Bound of \ref{u3} and \ref{u4}]
  Updates \ref{u3} and \ref{u4} are invoked at
  most $O(|E_\cP| N \log(d))$ times.
 \label{lem:u34}
\end{lemma}
%% \ifnum\final = 1

%% \begin{remark}\label{lowerbound}
%% In the full version we show the necessity of the finite precision by
%%  showing that our dependence on the discretization parameter $N$ is
%%  tight --- we show that given Assumption~\ref{assumpt:precision},
%%  there exists a polytope $\cP$ and a sequence of additional
%%  constraints such that any learning algorithm will make
%%  $\Omega(N\log(d))$ mistakes. This in particular implies that
%%  without any uniform upper bound on precision, it is impossible to
%%  learn with a finite mistake bound.

%% This lower bound motivates us to study a PAC-like variant of the
%%  problem, where the examples are not chosen in an adversarial manner,
%%  but instead are drawn independently at random from an arbitrary
%%  unknown distribution. We show that even if the constraints can be
%%  specified to arbitrary precision, there is an efficient learner that
%%  requires sample complexity only linear in the number of edges of the
%%  unknown constraint polytope.
%% \end{remark}
%% \else
\subsection{Necessity of the Precision Bound} \label{sec:lowerbound}
\snote{this section is now a remark in NIPS}
\rynote{Good!}
We show the necessity of
Assumption~\ref{assumpt:precision} by showing that the dependence on
the precision parameter $N$ in our mistake bound is tight.~We show that
subject to
Assumption \ref{assumpt:precision}, there exist a polytope  and a sequence of additional constraints such that any learning algorithm
will make $\Omega(N)$ mistakes. This implies that without any upper bound on precision,
it is impossible to learn with finite mistakes.
%\snote{reworked; I think there might be dependence on $d$ in the
%  bound, since the precision is on \emph{every} entry of the
 % points. we should make it consistent with the assumption
 % imposed}\rynote{Added $\log(d)$ dependence}

%\ifnum\final=1
%\begin{wrapfigure}{r}{0.4\textwidth}
%\begin{center}
%\includegraphics[width=0.3\textwidth]{3D}
%\end{center}
%\caption{The underlying polytope in the proof of Theorem~\ref{thm:impossibility}. The two learned edges  are in bold. }
%%which the learner guesses points on but does not know exactly where the feasibility region ends and the infeasibility region begins.}
%\label{fig:3d}
%\end{wrapfigure}
%\fi
%
\begin{theorem}
\label{thm:lower}
For any learning algorithm $\cA$ in the Known Objective Problem and any $d \geq 3$, 
there exists a polytope $\cP$ and a sequence of additional constraints
$\{\cN\utt\}_t$ such that the number of mistakes
made by $\cA$ is at least $\Omega(N)$.~\footnote{
We point out that the condition $d \geq 3$ is necessary in the statement of Theorem~\ref{thm:impossibility} since there
exists learning algorithms for $d=1$ and $d=2$ with finite mistake bounds independent of $N$.
\ifnum\cm=1
See the supplementary material.
\else
See Appendix~\ref{sec:1-2d}.
\fi
}
%% If $N$ bits of precision are required to represent the optimal points,
%% the number of mistakes of any learning algorithm is at least
%% $\Omega(N)$ for $d\ge 3$.
\label{thm:impossibility}
\end{theorem}

 \ifnum\final=0
\subsection{Stochastic Setting} \label{sec:stochastic}
\ifnum\final=0
Given the lower bound in Theorem~\ref{thm:impossibility}, we
ask ``In what settings we can still learn without an upper bound on the precision to which constraints are
specified?'' The lower bound  implies we must abandon the adversarial
setting so we consider a PAC style variant.
\else
Given the lower bound mentioned above, we now ask in what settings we
can still learn when we do not have a uniform upper bound on the
precision to which constraints can be specified.
\fi
In this variant, the additional constraint at
each day $t$ is drawn i.i.d. from some fixed but unknown
distribution $\cD$ over $\R^d\times
\R$ such that each point $(\bp, b)$ drawn from $\cD$ corresponds to
the halfspace $\cN=\{ \bbx \mid \bp \cdot \mathbf{x} \leq b\}$.
We make no assumption on the form of $\cD$ and require our bounds to hold in the worst case over all choices of $\cD$.
%\begin{definition}[Known Objective Stochastic Problem]
%Let an adversary choose an objective function $\bbc$ and a distribution $\cD$. We assume $\bbc$ is known and $\cN\utt$ is chosen i.i.d. from $\cD$ each day and will be revealed to the learner.
%With Assumptions \ref{assumpt:vertex} and \ref{assumpt:poly}, the \emph{Known Objective Stochastic Problem} considers a learner that wants to predicts $\bbx\utt = \myargmax_{\bbx \in \cP\utt} \bbc\cdot\bbx$ with the fewest number of mistakes.
%\end{definition}
%\sj{the same issue in the definition of mistake bound.}

We describe $\learnhull$ an algorithm based on the following high level idea: 
$\learnhull$ keeps
track of the convex hull $\cC^{(t-1)}$ of all the solutions observed up to day $t$. 
$\learnhull$ then behaves as if this convex hull is the entire feasible region. 
So at day $t$, given the constraint $\cN\utt = \{ \bbx \mid \bp^{(t)}\cdot \bbx\leq b\utt \}$, $\learnhull$ predicts $\hat\bbx\utt$ where
\rynote{Commented out $\hat\cC\utt$ notation because we never use it later.}
\ifnum\cm=1
$
\hat \bbx^{(t)} = {\argmax}_{\bbx \in \cC^{(t-1)}\cap \cN\utt} \enspace\bbc\cdot \bbx.
\label{eq:convex_guess}
$
\else
\begin{equation}
\hat \bbx^{(t)} = {\argmax}_{\bbx \in \cC^{(t-1)}\cap \cN\utt} \enspace\bbc\cdot \bbx.
\label{eq:convex_guess}
\end{equation}
\fi

$\learnhull$'s hypothetical feasible region is therefore always a subset of the true feasible region -- i.e. it can never make a mistake because its prediction was infeasible, but only because its prediction was sub-optimal. Hence, whenever $\learnhull$ makes a mistake, it must have observed a point that expands the convex hull. Hence, whenever it fails to predict $\bbx^{(t)}$,
$\learnhull$ will enlarge its feasible region by adding the point
$\bbx^{(t)}$ to the convex hull:
\ifnum\cm=1
$
\cC^{(t)} \gets \mathrm{Conv}(\cC^{(t-1)} \cup \{\bbx^{(t)}\}),
\label{eq:new_convex}
$
\else
\begin{equation}
\cC^{(t)} \gets \mathrm{Conv}(\cC^{(t-1)} \cup \{\bbx^{(t)}\}),
\label{eq:new_convex}
\end{equation}
\fi
otherwise it will simply set $\cC^{(t)} \gets \cC^{(t-1)}$ for the next day. 
\ifnum\cm=0
\ifnum\alt=0
$\learnhull$ is described formally in Algorithm~\ref{alg:cvxhull}.

\ifnum\alt=0
\begin{algorithm}[h]
 \begin{algorithmic}
\Procedure{$\learnhull$}{$\cD$}
\State $\cC^{(0)} \gets \emptyset.$  \ifnum\alt=0\hspace{90mm}\else\hspace{10mm}\fi$\triangleright$ Initialize
%\State\For{$t = 1, 2, \ldots$}
\State Observe $\cN\utt~\sim D$ and set $\hat{\bbx}\utt$ as in (\ref{eq:convex_guess}). \ifnum\alt=0\hspace{48mm}\else\hspace{10mm}\fi$\triangleright$ Predict
\State Observers $\bbx\utt = \myargmax_{\bbx\in \cP\cap\cN\utt} \bbc\cdot\bbx$ and update $\cC\utt$ as in (\ref{eq:new_convex}). \ifnum\alt=0\hspace{13mm}\else\hspace{10mm}\fi$\triangleright$ Update
%\EndFor
\EndProcedure
\end{algorithmic}
\caption{Stochastic Procedure ($\learnhull$)}
\label{alg:cvxhull}
\end{algorithm}
\fi
\fi
\fi
We show that the
expected number of mistakes of $\learnhull$ over $T$ days is linear in the number of edges of $\cP$ and only logarithmic in $T$.
\footnote{$\learnhull$ can be implemented efficiently in time $\poly(T, N, d)$ if all of the coefficients in the unknown constraints in $\cP$ are represented in $N$ bits.
Note that given the 
observed solutions so far and a 
new point, a separation oracle can be implemented in time $\poly(T, N, d)$ using a LP solver.
}
%First, in Theorem~\ref{thm:mistakebdd} we show that the expected number of mistakes made by $\learnhull$ scales logarithmically with the number of days and linearly in the number of edges of the polyhedron.
\begin{theorem}
\label{thm:mistakebdd}
For any $T > 0$ and any constraint distribution $\cD$, the expected
number of mistakes of $\learnhull$ after $T$ days
is bounded by $O\left(|E_{\cP}| \log(T) \right)$.
%\ar{Earlier in the document we seem to be using the notation $E_{\cP}$ for the edge set of the polytope now. We should make sure this is propagated throughout.}
%\sj{trying to fix them in this section.}
\end{theorem}
\ifnum\final=0
To prove Theorem~\ref{thm:mistakebdd}, first in Lemma~\ref{lem:stoch_case} we bound the
probability that the solution observed at day $t$ falls outside of the
convex hull of the previously observed solutions. 
This is the only event that can cause $\learnhull$ to make a mistake.
In Lemma~\ref{lem:stoch_case}, we abstract away the fact that the point observed at each day 
is the solution to some optimization problem. 
%-- Lemma~\ref{lem:stoch_case} holds for any distribution over points
% restricted to the edges of a polytope. 
%\sj{be sure to make things consistent now that the definition of $X$ does not include the point $x_t$.}
 \begin{lemma}
   Let $\cP$ be a polytope and $\cD$
   a distribution over points on $E_{\cP}$. Let 
$X=\{x_1,
   \ldots, x_{t-1}\}$
    be $t-1$ i.i.d. draws from $\cD$ and $x_t$ an additional independent draw from $\cD$. Then %% the
   %% probability that $x_t$ is an extreme point of $\mathrm{Conv}(X)$ is bounded by $2|E| / t$,
$
   \Pr[x_t\not\in \mathrm{Conv}(X)] \leq 2|E_{\cP}|/t
   $
   where the probability is taken over the draws of points 
$x_1,\ldots,x_t$
from $\cD$. 
   \label{lem:stoch_case}
 \end{lemma}
%We now present the proof of Theorem~\ref{thm:mistakebdd}.\\
\fi

Finally in Theorem~\ref{thm:stoch} we convert the bound on the expected number of mistakes 
of $\learnhull$ in Theorem~\ref{thm:mistakebdd} 
to a high probability bound. \footnote{$\learnedge$ fails to give any non-trivial mistake bound
in the adversarial setting.}

\begin{theorem}
\label{thm:stoch}
\rynote{We never use Known Objective Stochastic Problem, so just deleted it here.}
There exists a deterministic procedure %for the Known Objective Stochastic Problem 
such that after $T = O\left(|E_{\cP}|\log\left(1/\delta\right)\right)$ days, the probability (over the randomness of the additional constraint) that the procedure makes a mistake on day $T+1$ is at most $\delta$ for any $\delta \in (0,1/2)$. 
%\ar{Use new notation for the edge set}

%Let $\beta, \delta \in (0,1/6)$. Then after $T
%= \lceil2|E|/\beta\cdot\log\left(1/\delta\right)\rceil$ days,
%with probability at least $1-\delta$ over the draws of the previous
%$T$ constraints, the probability that $\learnhull$ makes a mistake on
%day $(T+1)$ is at most $\beta$, where $E$ denotes the set of edges
%in the polytope defined by the constraint $A\bbx\leq \bbb$.
\end{theorem}
%\ar{We are proving something different than what is claimed in the theorem, right? i.e. this is not a mistake bound on $\learnhull$, but on an algorithm that runs many copies of $\learnhull$ and aggregates the results. We should be clear about this in the theorem statement.}
\sj{better, now? Also, can somebody verify that my proof is correct?}
%\sj{kind of annoying because before this point we were thinking about mistakes after $T-1$ days. Am I too nit-picky or should we change this?}

\fi

\section{The Known Constraints Problem}\label{sec:known_constr}
We now consider the~\emph{Known Constraints Problem} in which the
learner observes the changing constraint polytope $\cP\utt$ at each day, but
does not know the changing objective function which we assume to be written as
$ \bbc\utt = \sum_{i \in S\utt} \bbv^i$, where $\{\bbv^i \}_{i \in
  [n]}$ are fixed but unknown.  Given $\cP\utt$ and the
subset $S\utt \subseteq [n]$, the learner must make a prediction
$\hat\bbx\utt$ on each day.  \sj{should we introduce the notion $[n]$
  or avoid using it as in the case in the introduction.}  \rynote{I
  put it in intro}\ifnum\final=0
%\subsection{An Ellipsoid Based Algorithm}
\fi
Inspired by~Bhaskar et al.~\cite{BLSS14}, we use the Ellipsoid algorithm to learn the coefficients $\{\bbv^i \}_{i \in
  [n]}$, and show that the mistake bound of the resulting algorithm is bounded by the (polynomial) running time of the Ellipsoid.  
  We use $V \in \R^{d\times n}$ to denote the matrix whose columns are $\bbv^i$ and make the following assumption on $V$.
\begin{assumpt}
  Each entry in $V$ can be written with $N$ bits of precision. Also w.l.o.g. $||V||_F \leq 1$.%% \footnote{Note that the choice of $1$ is arbitrary and can be replace by any constant.}
\label{assumpt:obj_precision}
\end{assumpt}

\iffalse
the objective $\bbc\utt$ each day to guarantee that the ellipsoid algorithm will terminate after only a polynomial number of days
(see more details in Theorem~\ref{thm:unknwon-obj}).

All the coefficients of the $n$ objective functions can be represented by a matrix with dimension $n\times d$.  We write $v^*$ to denote this matrix \ar{Why introduce new notation? i.e. how is $v^*$ different than $c^1,\ldots,c^n$?}. Furthermore, we assume the $i$th row of $v^*$
denoted by $v^*_i$ corresponds to $c^i$.
\fi
%\begin{assumpt}
%We assume that $\bbv^*$ can be represented as a multiple of $1/2^N$. Furthermore, without loss of generality, we assume $||\bbv^*||_F \leq 1$.\footnote{Note that the choice of $1$ is arbitrary and can be replace by any constant.}
%\label{assumpt:obj_precision}
%\end{assumpt}

%We now formally define the problem we consider in this section:
%\begin{definition}[Known Constraints Problem]
%Let an adversary choose an array $\bbv^* \in \R^{n \times d}$ that satisfies Assumption \ref{assumpt:obj_precision}.  Each day the adversary chooses a subset $S\utt \subseteq [n]$ and a polytope $\cP\utt \subset \R^d$ that obeys Assumption \ref{assumpt:precision}, which are both known to the learner.  After the learner makes a prediction, the adversary reveals the solution  $\bbx\utt \in \myargmax_{\bbx\in\cP\utt}\sum_{i \in S\utt} \sum_{j \in [d]} v_{ij}^* x_j$, satisfying Assumption \ref{assumpt:vertex}.   The learner's goal is to minimize the number of prediction mistakes made in the worst case over the choices of the adversary.
%\end{definition}
Similar to Section~\ref{sec:known_obj} 
we assume the coordinates of $\cP\utt$'s vertices can be written
with finite precision.\footnote{We again point out that this is implied if the halfspaces defining the
polytope are described with finite precision~\cite{GLS93}.}
\begin{assumpt}
The coordinates of each vertex of $\cP\utt$ 
can be written with $N$ bits of precision.
\label{assumpt:precision2}
\end{assumpt}

%Given a subset $S \subseteq [n]$, and a polytope $\cP$, let $\bbx^{S,\cP}$ denote the optimal solution to the instance defined by $S$ and $\cP$.  We want to define a feasibility region $\cF$ such that any point $W = (\bbw^1,\ldots,\bbw^n) \in \cF$ will make $\sum_{i \in S} \bbw^i \cdot \bbx^{S,\cP} > \sum_{i \in S} \bbw^i \cdot \bbx$ for any other $\bbx \in \cP$ where $S$ and $\cP$ can be chosen by an adversary.  This will ensure that if we find a point $W \in \cF$, then with the set $S\utt$ and polytope $\cP\utt$ given by the adversary each day, we need to only predict the point $\bbx \in \cP\utt$ with highest objective $\sum_{i \in S\utt} \bbw^i\cdot \bbx$ to ensure we never make a mistake.

We first observe that the coefficients of the objective function represent a point that is guaranteed to lie in a region $\cF$ (described below) which may be written as the intersection of possibly infinitely many halfspaces. Given a subset $S \subseteq [n]$ and a polytope $\cP$, let $\bbx^{S,\cP}$ denote the optimal solution to the instance defined by $S$ and $\cP$. Informally, the halfspaces defining $\cF$ ensure that for any problem instance defined by arbitrary choices of $S$ and $\cP$, the objective value of the \emph{optimal} solution $x^{S,\cP}$ must be at least as high as the objective value of any \emph{feasible} point in $\cP$.
Since the convergence rate of the Ellipsoid algorithm depends on the precision to which constraints are specified, we do not in fact consider a hyperplane for every feasible solution but only for those solutions that are vertices of the feasible polytope $\cP$. This is not a relaxation, since LPs always have vertex-optimal solutions.

We denote the set of all vertices of polytope $\cP$ by vert$(\cP)$, and the set of polytopes $\cP$ satisfying Assumption \ref{assumpt:precision2} by $\Phi$.  We then define $\cF$ as follows:
%\sj{I don't understand what the set $\cF$ is. Why it can't be a subset of $R^{n\times d}$ with bounded norm and precision?}
%\rynote{Does this make more sense?  This is just to write down a set of linear constraints to motivate why the separating hyperplanes in fact violates one of the constraints of our feasible region $\cF$ we are trying to find a point in.}
%\begin{equation}
%\cF = \left\{ \bbv \in \R^{n\times d} : \forall S \subset [n], \forall \cP  \in \Phi \quad \sum_{i \in S} \sum_{j \in [d]} v_{i,j} \left(x_j^{S,\cP} - x_j\right) \geq 0,\quad \forall \bbx \in \text{vert}(\cP) \right\}
%\label{eq:obj_feasible}
%\end{equation}
\ifnum\cm=1
\begin{align*}
\cF = \left\{ W = (\bbw^1,\ldots,\bbw^n) \in \R^{n\times d}\mid \forall S \subseteq [n], \forall \cP  \in \Phi, \sum_{i \in S} \bbw^i\cdot\left(\bbx^{S,\cP} - \bbx\right) \geq 0, \forall \bbx \in \text{vert}(\cP) \right\}
\label{eq:obj_feasible}
\end{align*}
\else
\begin{equation}
\cF = \left\{ W = (\bbw^1,\ldots,\bbw^n) \in \R^{n\times d}\mid \forall S \subseteq [n], \forall \cP  \in \Phi, \sum_{i \in S} \bbw^i\cdot\left(\bbx^{S,\cP} - \bbx\right) \geq 0, \forall \bbx \in \text{vert}(\cP) \right\}
\label{eq:obj_feasible}
\end{equation}
\fi
The idea behind our $\learnell$ algorithm is that we will run a copy of the Ellipsoid algorithm with variables $\bbw \in \R^{d\times n}$, as if we were solving the feasibility LP defined by the constraints defining $\cF$. We will always predict according to the centroid 
%\rynote{I am seeing ``center" not centroid when ellipsoid algorithm is mentioned.  Are they the same?}\ar{Think centroid is right, but am not sure} 
of the ellipsoid maintained by the Ellipsoid algorithm (i.e. its candidate solution). Whenever a mistake occurs, we are able to find one of the constraints that define $\cF$ such that our prediction violates the constraint -- exactly what is needed to take a step in solving the feasibility LP. Since we know $\cF$ is non-empty (at least the true objective function $V$ lies within it) we know that the LP we are solving is feasible. Given the polynomial convergence time of the Ellipsoid algorithm, this gives a polynomial mistake bound for our algorithm.

The Ellipsoid algorithm will generate a sequence of ellipsoids with decreasing volume such that each one contains feasible region $\cF$.
%the point representing the true objective function.
Given the ellipsoid $\cE\utt$ at day $t$, $\learnell$ uses the centroid
%\rynote{I am seeing ``center" not centroid when ellipsoid algorithm is mentioned.  Are they the same?}\ar{We could use center if that seems more common...}  
of $\cE\utt$ as its hypothesis for the objective function $W\utt = \left( (\bbw^1)\utt, \ldots, (\bbw^n)\utt\right) $.
Given the subset $S\utt$ and polytope $\cP\utt$, $\learnell$ predicts
%\begin{equation}
%\hat\bbx\utt \in \argmax_{\bbx \in \cP\utt} \left\{ \sum_{i \in S\utt} \sum_{j \in [d]}  \hat v\utt_{i,j} x_j\right\}.
%\label{eq:obj_guess}
%\end{equation}
\ifnum\cm=1
$
\hat\bbx\utt \in \myargmax_{\bbx \in \cP\utt} \{ \sum_{i \in S\utt}  (\bbw^i)\utt \cdot \bbx\}.
\label{eq:obj_guess}
$
\else
\begin{equation}
\hat\bbx\utt \in \myargmax_{\bbx \in \cP\utt} \{ \sum_{i \in S\utt}  (\bbw^i)\utt \cdot \bbx\}.
\label{eq:obj_guess}
\end{equation}
\fi
When a mistake occurs, $\learnell$ finds the hyperplane 
\ifnum\cm=1
$
\cH\utt = \left\{W = (\bbw^1,\ldots, \bbw^n) \in \R^{n\times d} : \sum_{i \in S\utt}\bbw^i \cdot (\bbx\utt - \hat\bbx\utt ) > 0 \right\}
\label{eq:sep_hyp}
$
\else
\begin{equation}
\cH\utt = \left\{W = (\bbw^1,\ldots, \bbw^n) \in \R^{n\times d} : \sum_{i \in S\utt}\bbw^i \cdot (\bbx\utt - \hat\bbx\utt ) > 0 \right\}
\label{eq:sep_hyp3}
\end{equation}
\fi
that separates the centroid of the current ellipsoid (the current candidate objective) from $\cF$.

After the update, we use the Ellipsoid algorithm to compute 
the minimum-volume ellipsoid $\cE\uttf$ that contains
$\cH\utt \cap \cE\utt$. On day $t+1$, $\learnell$ sets
$W\uttf$ to be the centroid of
$\cE\uttf$.  \ifnum\final=1 We call this learning procedure
$\learnell$ with a formal description in the full version. \fi
\ifnum\cm=0 
\ifnum\alt=0
The above procedure is formalized in Algorithm~\ref{alg:learnellip}.
\begin{algorithm}[h]
 \begin{algorithmic}
 \Procedure{$\learnell$}{$\cA$} \hspace{49mm} $\triangleright$ Against adversary $\cA$\\
	\State $\cI^{(1)} \gets \left(\cE^{(1)} = \{\bbz \in \R^{n\times d} : ||\bbz||_F\leq 2\}, W^{(1)} = \cent(\cE^{(1)})\right)$
	%\mathbf{1} .$
	\hspace{10mm} $\triangleright$ Initialize
	\For{$t = 1 \ldots $}\\
	 \hspace{10mm} Given $S\utt, \cP\utt$, set $\hat\bbx\utt$ as in (\ref{eq:obj_guess}).
	 \hspace{56mm} $\triangleright$ Predict
	 \If{$\bbx\utt \neq \hat\bbx\utt$}\\
		\hspace{15mm} set $\cH\utt$ as in \eqref{eq:sep_hyp3}.
		\hspace{73mm} $\triangleright$ Update\\
		\hspace{15mm}$\cE\uttf \gets \ellip(\cH\utt,\cE\utt), W^{\uttf}=\cent(\cE\uttf).$\\
		\hspace{15mm}$\cI\uttf = \left(\cE\uttf, W^{\uttf}\right).$
		\Else\\
		\hspace{15mm}$\cI\uttf \gets \cI\utt.$
\EndIf

\EndFor
\EndProcedure
\end{algorithmic}
\caption{Learning with Known Constraints ($\learnell$)}
\label{alg:learnellip}
\end{algorithm}
\fi
\fi

\ifnum\final=0
%\subsection{Analysis of $\learnell$}
\fi
We  left the procedure used to solve the LP in the prediction rule of $\learnell$ unspecified.  To simplify our analysis, 
we use a specific LP solver to obtain a prediction $\hat\bbx\utt$ which is a vertex of $\cP\utt$.
\ifnum\cm=0 We defer all the proofs of this section to Appendix~\ref{sec:missingproof4}.\fi
\begin{theorem}[Theorem 6.4.12 and Remark 6.5.2~\cite{GLS93}]
There exists a LP solver that runs in time polynomial in the length of its input and returns an exact solution that is a vertex of $\cP\utt$.
\label{thm:lin_solve}
\end{theorem}
In Theorem~\ref{thm:unknwon-obj}, we show that the number of mistakes made by $\learnell$ is at most the number 
of updates that the Ellipsoid algorithm makes before it finds a point in $\cF$ and the number of updates of the Ellipsoid algorithm can be bounded 
by well-known results from the literature on LP. 
%\ifnum\alt=0 We defer these details and the proof of Theorem~\ref{thm:unknwon-obj} to Appendix~\ref{sec:missingproof4}.\fi
\begin{theorem}
\label{thm:unknwon-obj}
The total number of mistakes and the running time of $\learnell$ in the Known Constraints Problem is at most  $\poly(n,d,N)$.

%\rynote{No time for exact bound}
\end{theorem}

\newpage
\bibliographystyle{acm}
\bibliography{./references}

\newpage
\appendix
\section{Polynomial Mistake Bound with Exponential Running Time}
\label{sec:poly-mistake}
In this section we give a simple randomized algorithm for the unknown constraints problem, that in expectation makes a number of mistakes that is only linear in the dimension $d$, the number of rows in the unknown constraint matrix $A$ (denoted by $m$), and the bit precision $N$, but which requires exponential running time. When the number of rows is large, this can represent an exponential improvement over the mistake bound of $\learnedge$, which is linear in the number of \emph{edges} on the polytope $\cP$ defined by $A$.
This algorithm which we describe shortly is a randomized variant of the well known halving algorithm~\cite{Littlestone87}. 
We leave it as an open problem whether the mistake bound achieved by this algorithm can also be achieved by a computationally efficient algorithm.

Let $\cK$ be the hypothesis class of all polytopes formed by $m$ constraints in $d$ dimensions, such that each entry of each constraint can be written as a multiple of $1/2^N$ (and without loss of generality, up to scaling, has absolute value at most 1).  We then have
$$
|\cK| = 2^{O(d m N)}.
$$
We write $\cK\utt$ to denote the polytopes that are consistent with the examples and solutions we have seen up to and including day $t$.  Note that $|\cK\utt| \geq 1$ for every $t$ because there is some polytope (specifically the true unknown polytope $\cP$) that is consistent with all the optimal solutions.  On each day $t$ we keep track of consistent polytopes and more specifically update the set of consistent polytopes by \ar{The optimization below looked wrong -- it didn't mention the new constraint each day. I changed it from optimizing over $\cP$ to $\cP\utt$.}\rynote{Good!}
\begin{equation}
\cK\uttf = \left\{\cP \in \cK\utt \mid \bbx\utt \in \myargmax_{\bbx \in \cP\cap \cN\utt} \bbc\cdot \bbx \right\},
\label{eq:cons_poly}
\end{equation}
where $\cN\utt$ is the new constraint on day $t$. The formal description of the algorithm, $\texttt{FCP}$, is presented in Algorithm~\ref{alg:FCP}.
To predict at each day, $\texttt{FCP}$ selects a polytope $\hat\cP\utt$ from $\cK\utt$ uniformly at random and guesses $\hat\bbx\utt$ that solves the following LP:
$\max_{\bbx\in\hat\cP\utt\cap\cN\utt}\bbc\cdot\bbx$.
\begin{algorithm}[h]
 \begin{algorithmic}
 \Procedure{$\texttt{FCP}$}{}
 \State $\cK^{(1)} = \cK$.  \hspace{95mm} $\triangleright$ Initialize
 \For{$t = 1 \ldots $}
 \State Choose $\hat\cP\utt \in \cK\utt$ uniformly at random.
 \State Guess $\hat\bbx\utt \in \myargmax_{\bbx \in \hat\cP\utt\cap\cN\utt} \bbc\cdot\bbx$. \hspace{54mm}$\triangleright$ Predict
 \State Observe $\bbx\utt$ and set $\cK\uttf$ as in \eqref{eq:cons_poly}.
 \EndFor
\EndProcedure
\end{algorithmic}
\caption{Find Consistent Polytope $\texttt{FCP}$}
\label{alg:FCP}
\end{algorithm}

We now bound the expected number of mistakes that $\texttt{FCP}$ makes.
\begin{theorem}
The expected number of mistakes that $\texttt{FCP}$ makes is at most $\log(|\cK|) = O(dmN)$, where the expectation is over the 
randomness of $\texttt{FCP}$ and possible randomness of the adversary.
\end{theorem}
\begin{proof}
First note that the probability that $\texttt{FCP}$ \emph{does not} make a mistake at day $t$ can be expressed as 
$|\cK\uttf|/|\cK\utt|$.
This is because if $\texttt{FCP}$ makes a mistake at day $t$, it must have selected a polytope that will be eliminated at the next day (also note that $\texttt{FCP}$ selects its polytope from among the consistent set uniformly at random).
Now consider the product of these probabilities over all days $t=1\ldots T$.
$$
\prod_{t = 1}^T (1 - \Prob{}{\text{Mistake at day $t$}} ) = \prod_{t = 1}^T \frac{|\cK\uttf|}{|\cK\utt|} =  \frac{|\cK^{(T+1)} |}{|\cK^{(1)} |}.
$$
Finally, note that the expected number of mistakes is the sum of probabilities of making mistakes over all days.
Using the inequality $(1-x) \leq e^{-x}$ for every $x\in[0,1]$ and rearranging terms we get
$$
\sum_{t=1}^T \Prob{}{\text{Mistake at day $t$}} \leq \log\left( \frac{|\cK^{(1)} |}{|\cK^{(T+1)} |}\right) \leq O(dmN),
$$
since $|\cK^{(1)} | = 2^{O(d m N)}$ and $|\cK^{(T+1)} | \ge 1$.
\end{proof}

Finally, we remark that the randomized halving technique above will also result in a polynomial mistake
bound in the more demanding variant where not only the underlying constraint matrix but also the linear objective function is unknown.
This is because the coefficients of the objective function can be written in $dN$ bits if they are also represented 
with finite precision. However, the issue about the exponential running time 
still exists in the new setting.
%\ifnum\final=0
%\section{Proof of Theorem \ref{thm:impossibility}}
%\else
\section{Missing Proofs from Section~\ref{sec:known_obj}}
\label{sec:missingproof3}
%%%%%%%%%%%%%%%%%%%%%%%%%%%%%%%%%
%%%%% LEARN EDGE PROPERTIES%%%%%%%%%%%%%%
%%%%%%%%%%%%%%%%%%%%%%%%%%%%%%%%%
\subsection{Section~\ref{sec:learn-edge}}
\textbf{Proof of Lemma~\ref{lem:edge}.}
  Let $\bbx^*$ be the optimal solution of the linear program solved over the unknown polytope $\cP$, without the added constraint i.e. $\bbx^* \equiv \argmax_{x\in \cP} \bbc \cdot \bbx$.
  \begin{enumerate}
  \item
  Suppose that
  $\bbx^* \in \cN\utt$, then clearly $\bbx\utt = \bbx^*$. By Assumption $\ref{assumpt:vertex}$, $\bbx^*$ lies on a vertex of $\cP$ and therefore $\bbx\utt$ lies on one of the edges of $\cP$.
 \item
  Suppose that $\bbx^* \notin \cN\utt$ i.e. $\bp\utt \cdot \bbx^* > b\utt$. Then we
  claim that the optimal solution $\bbx\utt$ satisfies $\bp\utt \cdot
  \bbx\utt = b\utt$. Suppose to the contrary that $\bp\utt\cdot
  \bbx\utt < b\utt$. Since $\bbc \cdot \bbx^* \geq \bbc\cdot \bbx\utt$, then for any point $\bby \in \mathrm{Conv}(\bbx\utt, \bbx^*)$,
\begin{align*}
 \bbc\cdot \bby =  \bbc \cdot (\alpha\bbx\utt + (1-\alpha)\bbx^*) = \alpha (\bbc\cdot
  \bbx\utt) + (1 -\alpha) (\bbc \cdot \bbx^*) \geq \bbc \cdot \bbx\utt \qquad \forall \alpha \in [0,1].
\end{align*}
Since $\bbx\utt$ strictly satisfies the new constraint, there exists some point $\bby^*\in \mathrm{Conv}(\bbx\utt, \bbx^*)$ where $\bby^* \neq \bbx\utt$
such that $\bby^* \in \cP\utt$ (i.e. $\bby^*$ is also feasible).
It follows that $\bbc\cdot \bby^* \geq \bbc\cdot
\bbx\utt$, which contradicts
Assumption~\ref{assumpt:vertex}. Therefore, $\bbx\utt$ must bind the
additional constraint. Furthermore, by non-degeneracy Assumption \ref{assumpt:bind}, $\bbx\utt$ binds exactly $(d-1)$ constraints in $\cP$, i.e. $\bbx\utt$ lies at the intersection of $d-1$ hyperplanes of $\cP$ which are linearly independent by Assumption  \ref{assumpt:rank}. Therefore, $\bbx\utt$ must be on an edge of $\cP$. \ar{Why does the rank condition imply this? I'm not doubting it, but it should be spelled out.}\rynote{Better?}
\end{enumerate}
\qed

\textbf{Proof of Lemma~\ref{lem:collinear}.}
  Without loss of generality, let us assume $\bby$ can be written as convex combination of
  $\bbx$ and $\bbz$ i.e. $\bby = \alpha \bbx + (1 -
  \alpha)\bbz$ for some $\alpha\in (0,1)$. Let $B_y = \{j \mid A_j
  \bby = \bbb_j\}$ be the set of binding constraints for $\bby$. 
  We know that $|B_y| \geq d- 1$ by Assumption~\ref{assumpt:bind}. 
%  \textbf{\color{red} shahin: why? is this a result of Assumption 5 or 6?}
%  \textbf{\color{blue} ryan: It is a result of Assumption 6.  Each day, our solution is unique and must bind exactly $d$ constraints.  Those $d$ constraints are either from all the original polytope (which would be the actual optimal to the underlying program) or $d-1$ constraints from the original polytope and 1 from the new constraint.}
  For any $j$ in $B_y$, we consider
  the following two cases.
  \begin{enumerate}
  \item
  At least one of $\bbx$ and $\bbz$ belongs to the hyperplane
  $\{ \bbw \mid A_j \bbw =  \bbb_j\}$. 
%  \textbf{{\color{red} shahin: should this be half-space or the definition
%  should change $\leq$ with $=$? The same applies to case 2.}} 
%  \textbf{\color{blue}ryan: I think it should be $=$ in both.}
  Then we claim that all three points bind the same
  constraint. Assume that $A_j\bbx = \bbb_j$, then we must have
  \[
  A_j\bbz = \frac{A_j (\bby - \alpha\bbx)}{(1 - \alpha) } =
  \frac{\bbb_j - \alpha \bbb_j}{(1 - \alpha)} = \bbb_j.
  \]
  Similarly, if we assume $A_j\bbz = \bbb_j$, we will also have
  $A_j\bbx = \bbb_j$.
  \item
  None of  $\bbx$ and $\bbz$ belongs to the hyperplane
  $\{ \bbw \mid A_j \bbw = \bbb_j\}$ i.e. $A_j \bbx < \bbb_j$ and $A_j \bbz < \bbb_j$
  both hold. Then we can write
  \[
 \bbb_j =   A_j\bby = \alpha A_j \bbx + (1 - \alpha) A_j \bbz < \alpha \bbb_j + (1 - \alpha) \bbb_j = \bbb_j,
  \]
  which is a contradiction.
\end{enumerate}
  It follows that for any $j\in B_y$, we have $A_j\bbx = A_j\bby =
  A_j\bbz = \bbb_j$. Since $|B_y| \geq d - 1$, we know by
  Assumption~\ref{assumpt:rank} that the set of points that bind any set of $d-1$ constraints in $B_y$ will form an edge-space and further this edge-space will include $\bbx,\bby,$ and $\bbz$.  
  % constraints form an
  %edge of the polytope, and therefore all three points must lie on the
  %same edge of the polytope.  The line that contains these three points is then an edge-space of $\cP$.
\qed
\ifnum\final=0

\textbf{Proof of Lemma~\ref{lem:corner-case}.}
First, note that the observed solution $\bbx\utt$ is a vertex in the polytope
$\cP\utt = \cP\cap\tN\utt$, that is an intersection of \emph{exactly}
$d$ constraints by Assumption~\ref{assumpt:vertex} and
Assumption~\ref{assumpt:bind}. Second, note that all points in $e$ bind at
least $d - 1$ constraints in $\cP$ and since $e\subseteq \tN\utt$,
then all points in $e$ bind at least $d$ constraints in $\cP\utt$.
It follows that any vertex of $\cP\utt$ on $e$ must bind at least
$(d+1)$ constraints, which rules out the possibility of $\bbx\utt$
being on $e$.
\qed

\fi

%%% COMMENETED FOR COLT %%%
%%% COMMENETED FOR COLT %%%
%%% COMMENETED FOR COLT %%%
%%%%%%%%%%%%%%%%%%%%%%%%%%%%%%%%%
%%%%% MISTAKE BOUND ANALYSIS%%%%%%%%%%%%%
%%%%%%%%%%%%%%%%%%%%%%%%%%%%%%%%%
\subsection{Section~\ref{sec:edge_analysis}}
\ifnum\final=0
\textbf{Proof of Lemma~\ref{lem:u1}.}
  As soon as $\learnedge$ invokes update rule \ref{u1}, it records the
  solution $\bbx^* \equiv \argmax_{\bbx\in \cP} \bbc\cdot \bbx$. Then, the
  prediction rule specified by~\ref{p1} prevents further updates of
  this type.  This is because $\bbx^*$ continues to remain optimal if it feasible in the more
  constrained problem (optimizing over the polytope $\cP\utt$).
\qed
\fi
\ifnum\final=0

\textbf{Proof of Lemma~\ref{lem:u2}.}
  Rule $\ref{u2}$ is
  invoked only when $\bbx\utt \notin X\utt$ and $\bbx\utt\notin e$ for any of $e\in E\utt$. 
  So after each invokation, a new point on the edge of $\cP$ is observed.
  Whenever $3$ points are observed on the same edge of $\cP$, the edge-space
  is learned by Lemma~\ref{lem:collinear} (since the points are necessarily collinear). Hence,
  the total number of times rule $\ref{u2}$ can be invoked is at most
  $3|E_\cP|$.
\qed
\fi

We now introduce Lemmas~\ref{lem:mod1}~and~\ref{lem:mod2} that will be used in the proof of Lemma \ref{lem:half}
which itself will be useful in the proof of Lemma~\ref{lem:u34}. But first, for completeness, in Lemma~\ref{lem:exists_edge}
we show that we are guaranteed the existence of an edge-space if the update implemented is \ref{u3} or \ref{u4}.
\begin{lemma}
$\newline$
\begin{enumerate}
\item[(1)] If update rule \ref{u3} is used, then there exists
edge-space $\hat e \in E\utt$ such that $\hat\bbx\utt \in \hat
e$. 
\item[(2)] If update rule \ref{u4} is used, then there
exists edge-space $e \in E\utt$ such that $\bbx\utt \in e$.
\end{enumerate}
\label{lem:exists_edge}
\end{lemma}
\begin{proof}
We prove this by contradiction.
First consider the case in which $\bbc \cdot \hat\bbx\utt
> \bbc \cdot \bbx\utt$ and suppose $\hat\bbx\utt \in \{ \bbx \in X\utt\mid \forall e\in E\utt, \bbx \notin e\}$.  When this is the case we know
that $\hat\bbx\utt$ is feasible at day $t$ and this contradicts
$\bbx\utt$ being optimal at that day because $\bbc\cdot\hat\bbx\utt
> \bbc\cdot\bbx\utt$.

Next consider the case in which $\bbc \cdot \hat\bbx\utt
 < \bbc \cdot \bbx\utt$ and suppose $\bbx\utt \in \{\bbx \in X\utt\mid \forall
 e \in E\utt, \bbx \not\in e\}$.  We would have
 used \ref{p3} to make a prediction because
 $\tN\utt\cap\mathrm{Ext}\utt$ is non-empty and includes at least the
 point $\bbx\utt$.  Note that by \ref{p3}, we have $\hat\bbx\utt
 = \argmax_{\tN\utt\cap\mathrm{Ext}\utt} \bbc \cdot \bbx$.  Since
 $\bbx\utt \in \tN\utt\cap\mathrm{Ext}\utt$, we must also have
 $\bbc \cdot \hat\bbx\utt \geq \bbc \cdot \bbx\utt$, which is again a 
 contradiction.
\end{proof}

\begin{lemma}
If \ref{u3} is implemented at day $t$, then $\hat\bbx\utt \notin \cP$ and $\hat\bbx\utt \in  (Q_{\hat e}^i)\utt \cap \text{Ext}\utt $ for some $i = 0$ or $1$ where $\hat e$ is given in \ref{u3}. 
\label{lem:mod1}
\end{lemma}
\begin{proof}
Each time the algorithm makes update \ref{u3} we know that the
algorithm's prediction $\hat\bbx\utt$ was on some edge-space $\hat e \in E\utt$ by Lemma~\ref{lem:exists_edge}.  Therefore, $\learnedge$ did not use \ref{p1} or \ref{p2}
to predict $\hat\bbx\utt$. So we only need to check \ref{p3} and \ref{p4}.
\begin{itemize}
\item If \ref{p3} was used, we know that $\hat\bbx\utt \in \tN\utt$ but $\hat\bbx\utt$ must violate a constraint of $\cP$, due to $\bbx\utt$ being the observed solution and having lower objective value.  This implies that $\hat\bbx\utt$ is in some questionable region, say $(Q_{\hat e}^i)\utt$ for $i = 0$ or $1$ but also in the  extended feasible on $\hat e$, i.e. $\hat\bbx\utt \in \text{Ext}\utt  \cap (Q_{\hat e}^i)\utt$.  
\item If \ref{p4} was used, then $\text{Ext}\utt = \emptyset$. However $\learnedge$ selected $\hat\bbx\utt$ from $\text{Cand}\utt \neq \emptyset$ with the lowest objective value.  
Finally, when updating with \ref{u3} \emph{(i)} $\bbx\utt\in\text{Cand}\utt$ and \emph{(ii)} $\bbc\cdot \bbx\utt < \bbc \cdot\hat\bbx\utt$. So we could not have used \ref{p4} to predict $\hat\bbx\utt$.
\end{itemize}
\end{proof}

\begin{lemma}
If \ref{u4} is implemented at day $t$, then $\bbx\utt \in (Q_e^i)\utt \backslash \text{Ext}\utt$ for some $i = 0$ or $1$ where $e$ is given in \ref{u4}.  
\label{lem:mod2}
\end{lemma}
\begin{proof}
As in Lemma~\ref{lem:mod1}, $\learnedge$ did not use \ref{p1} or \ref{p2}
to predict $\hat\bbx\utt$ (again by application of Lemma~\ref{lem:exists_edge}). So we only need to check \ref{p3} and \ref{p4}.
\begin{itemize}
\item If \ref{p3} was used,  then $\learnedge$ did not guess $\bbx\utt$ which had the higher objective because it was outside of $\text{Ext}\utt$ along edge-space $e$.  Since $\bbx\utt$ is feasible, it must have been on some questionable region on $e$, say $(Q_e^i)\utt$ for some $i = 0$ or $1$.  Hence, $\bbx\utt \in (Q_e^i)\utt \backslash \text{Ext}\utt$.
\item If \ref{p4} was used, then $\text{Ext}\utt = \emptyset$ and thus $\bbx\utt$ was a candidate solution but outside of the extended feasible interval along edge-space $e$.  Further, because $\bbx\utt \in \cP$ we know that $\bbx\utt$ must be in some questionable interval along $e$, say $(Q_e^i)\utt$ for some $i = 0$ or $1$.  Therefore, $\bbx\utt \in (Q_e^i)\utt \backslash \text{Ext}\utt$.
\end{itemize}
\end{proof}
\ifnum\final=1
\ifnum\final=0
\begin{proof}
\else
\begin{proof}[Proof of Lemma~\ref{lem:edge}]
\fi
  Let $\bbx^*$ be the optimal solution of the linear program solved over the unknown polytope $\cP$, without the added constraint i.e. $\bbx^* = \argmax_{x\in \cP} \bbc \cdot \bbx$.
  \begin{enumerate}
  \item
  Suppose that
  $\bbx^* \in \cN\utt$, then clearly $\bbx\utt = \bbx^*$. By Assumption $\ref{assumpt:vertex}$, $\bbx\utt$ lies on a vertex of $\cP$, and therefore $\bbx\utt$ lies on one of the edges of $\cP$.
 \item
  Suppose that $\bbx^* \notin \cN\utt$ i.e. $\bp\utt \cdot \bbx^* > b\utt$. Then we
  claim that the optimal solution $\bbx\utt$ satisfies $\bp\utt \cdot
  \bbx\utt = b\utt$. Suppose to the contrary that $\bp\utt\cdot
  \bbx\utt < b\utt$. Since $\bbc \cdot \bbx^* \geq \bbc\cdot \bbx\utt$, then for any point $\bby \in $ conv$(\bbx\utt, \bbx^*)$,
\begin{align*}
 \bbc\cdot \bby =  \bbc \cdot (\alpha\bbx\utt + (1-\alpha)\bbx^*) = \alpha (\bbc\cdot
  \bbx\utt) + (1 -\alpha) (\bbc \cdot \bbx^*) \geq \bbc \cdot \bbx\utt \qquad \forall \alpha \in [0,1].
\end{align*}
Since $\bbx\utt$ strictly satisfies the new constraint, there exists some point $\bby^*\in $ conv$(\bbx\utt, \bbx^*)$ where $\bby^* \neq \bbx\utt$
such that $\bby^* \in \cP\utt$ (i.e. $\bby^*$ is also feasible).
It follows that $\bbc\cdot \bby^* \geq \bbc\cdot
\bbx\utt$, which contradicts
Assumption~\ref{assumpt:vertex}. Therefore, $\bbx\utt$ must bind the
additional constraint. Furthermore, by non-degeneracy Assumption \ref{assumpt:bind}, $\bbx\utt$ binds exactly $(d-1)$ constraints in $\cP$, i.e. $\bbx\utt$ lies at the intersection of $d-1$ hyperplanes of $\cP$. As long as non-degeneracy Assumption  \ref{assumpt:rank} holds we know that these hyperplanes must be linearly independent, so $\bbx\utt$ must be on an edge of $\cP$. \ar{Why does the rank condition imply this? I'm not doubting it, but it should be spelled out.}\rynote{Better?}
\end{enumerate}
\end{proof}

\textbf{Proof of Lemma~\ref{lem:corner-case}.}
First, note that the observed solution $\bbx\utt$ is a vertex in the polytope
$\cP\utt = \cP\cap\tN\utt$, that is an intersection of \emph{exactly}
$d$ constraints by Assumption~\ref{assumpt:vertex} and
Assumption~\ref{assumpt:bind}. Second, note that all points in $e$ bind at
least $d - 1$ constraints in $\cP$ and since $e\subseteq \tN\utt$,
then all points in $e$ bind at least $d$ constraints in $\cP\utt$.
It follows that any vertex of $\cP\utt$ on $e$ must bind at least
$(d+1)$ constraints, which rules out the possibility of $\bbx\utt$
being on $e$.
\qed

\fi

\begin{lemma}
Each time \ref{u3} or \ref{u4} is used, there is a questionable interval on some edge-space whose length is decreased by at least a factor of two.  
\label{lem:half}
\end{lemma}
\ifnum\final=0
\begin{proof}
From Lemma \ref{lem:mod1} we know that if \ref{u3} is used then $\hat\bbx\utt \in \hat e$, is infeasible but outside of the known infeasible interval $(Y_{\hat e}^i)\utt$ and inside of the extended feasible interval along $\hat e$.  Note that if a point $\bbx$ is infeasible along edge space $\hat e$ in the questionable interval $(Q_{\hat e}^i)\utt$, then the constraint it violates is also violated by all points in $Y_e^i$. Hence the interval $\mathrm{Conv}(\bbx,(Y_e^i)\utt)$ contains only infeasible points. By the definition of $(M_e^i)\utt$ and the fact that $\hat \bbx\utt$ is in the extended feasible region on $\hat e$, we know that 
$$
|(Q_e^i)\uttf|= \left| (Q_e^i)\utt \backslash \mathrm{Conv} \left(\hat\bbx\utt, (Y_e^i)\utt \right)\right| \leq \left|(Q_e^i)\utt \backslash\mathrm{Conv} \left((M_e^i)\utt, (Y_e^i)\utt \right) \right|= \frac{|(Q_e^i)\utt|}{2}.
$$
Further, from convexity we know that if $\bbx\utt$ is feasible on edge-space $e$ at day $t$, then the interval $\mathrm{Conv}(\bbx\utt,F_e\utt)$ only contains feasible points on $e$.  We know that $\bbx\utt$ is feasible and in a questionable interval $(Q_e^i)\utt$ along edge space $e$ but outside its extended feasible region, by  Lemma \ref{lem:mod2}.  Thus, by definition of the midpoint $(M_e^i)\utt$ we have
$$
|(Q_e^i)\uttf|= \left| (Q_e^i)\utt \backslash \mathrm{Conv}\left(\hat\bbx\utt, (F_e)\utt \right)\right| \leq \left|(Q_e^i)\utt \backslash\mathrm{Conv} \left((M_e^i)\utt, F_e\utt \right) \right|= \frac{|(Q_e^i)\utt|}{2}.
$$
\end{proof}

\fi

\textbf{Proof of Lemma~\ref{lem:u34}.}
  Let $Q_e^i$ be the updated questionable interval. We know initially
  $Q_e^i$ has length at most than $2\sqrt{d}$ by Assumption~\ref{assumpt:poly}. 
  In Lemma~\ref{lem:half} we showed that each time an
  update~\ref{u3} or~\ref{u4} is invoked, the length of $Q_e^i$ is decreases by at least a half.
  Then after at most $O(N\log(d))$ updates, the
  interval will have length less than $2^{-N}$ after which the interval will be updated at most once
  because there is at most one point up to precision $N$ in it.
%  when
%  $\elim$ is invoked to eliminate the interval. 
  
  Therefore, the
  total number of updates on $Q_e^i$ is bounded by $O(N\log(d))$. 
  Since there are at most $2|E_\cP|$ questionable intervals, the total number of updates~\ref{u3} and~\ref{u4} is bounded by $O(|E_\cP|N\log(d))$.
\qed

%%%%%%%%%%%%%%%%%%%%%%%%%%%%%%%%%
%%%%%LOWER BOUND %%%%%%%%%%%%%%%%%%%
%%%%%%%%%%%%%%%%%%%%%%%%%%%%%%%%%
\subsection{Section~\ref{sec:lowerbound}}
%\ifnum\final=0
%\begin{proof}
%\else
%\begin{proof}[Proof of Theorem \ref{thm:impossibility}]
%\fi
We prove the lower bound in Theorem \ref{thm:impossibility} initially for $d=3$.
\begin{theorem}
\label{thm:imp-3d}
If Assumptions \ref{assumpt:vertex} and \ref{assumpt:precision} hold, then the number of mistakes of any learning algorithm 
in the known objective problem
is at least $\Omega(N)$ for $d=3$.
\end{theorem}
\begin{proof}
The high level idea of the proof is as follows. In each day the adversary can pick two points on the two bold edges in Figure~\ref{fig:3d}
as the optimal points and no matter what the learner predicts, the adversary can return a point that is different than
the guess of the learner as the optimal point. 
If the adversary picks the midpoint of the questionable region in each day, then the size of the questionable region in both of the lines
will shrink in half. So this process can be repeated $N$ times where each entry of every vertex can be written with as a multiple of $1/2^N$, by Assumption \ref{assumpt:precision}.
Finally, we show that at the end of this process, the adversary can return a simple polytope 
 which is consistent with all the observed optimal points so far. 
 
We formalize this high level in procedure $\adversary$ that takes as input any learning 
algorithm $\cL$ and interacts with $\cL$ for $N$ days. Each day the adversary presents a constraint.
Then no matter what $\cL$ 
predicts, the adversary ensures that $\cL$'s prediction is incorrect.  
After $N$ interactions, the adversary outputs a feasible polytope that is 
consistent with all of the previous actions of the adversary.

\begin{figure}[h]
 \begin{minipage}[c]{.6\textwidth}
\centering
\includegraphics[width=0.7\textwidth]{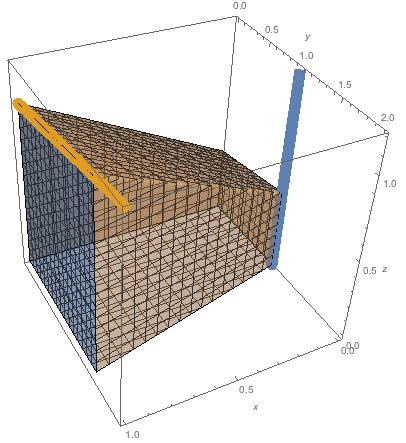}
\end{minipage}
\begin{minipage}[c]{.34\textwidth}
\caption{The underlying polytope in the proof of Theorem~\ref{thm:impossibility}. The two learned edges  are in bold. }
%which the learner guesses points on but does not know exactly where the feasibility region ends and the infeasibility region begins.}
\label{fig:3d}
\end{minipage}
\end{figure}

%%%FIX
%\begin{algorithm}[h]
% \KwIn{Any learning algorithm $\cL$ and bit precision $N$}
% \KwOut{ Polytope $\cP$ that is consistent with $\cL$ making a mistake each day.}
% {\SetAlgoNoLine
% \begin{algorithmic}
% \Procedure{$\adversary$}{$\cL,N$} 
%\State Set $R_1^{(0)}=[0,1], R_2^{(0)}=[1,2]$. \hspace{78mm} $\triangleright$ Initialize
%	\State\For{$t = 1,\cdots, N$}{		
%		\State $ \left((\bp\utt,q\utt),r_1\utt,r_2\utt )\gets \nac(R_1\uttb , R_2\uttb\right)$.
%		\State  Show constraint $\bp\utt\cdot\bbx \leq q\utt$ to $\cL$.  \hspace{64mm} $\triangleright$ Constraint
%		\State  Get prediction $\hat\bbx\utt$ from $\cL$. 
%		\State$ \left(\bbx\utt, R_1\utt, R_2\utt\right)\gets\ad\left(R_1\uttb , R_2\uttb, r_1\utt,r_2\utt, \hat\bbx\utt\right)$. \hspace{41mm} $\triangleright$ Update
%		\State Adversary reveals the optimal point $\bbx\utt\ne\hat\bbx\utt$ and updates regions $R_1\utt$ and $R_2\utt$.}
%	\State $A,\bbb \gets \conmatr(R_1^N,R_2^N)$ \hspace{28mm} $\triangleright$ Constraint matrix consistent with $\{ \bbx\utt\mid t \in [N]\}$ \\
%	\Return $A,\bbb$
% \EndProcedure
% \end{algorithmic}
% }
% \caption{Adversary Updates ($\adversary$)}
% \label{alg:adversary}
%\end{algorithm} 

In procedure $\adversary$, subroutines $\nac$ and $\ad$ are used 
to pick a constraint and return an optimal point that causes $\cL$ to make a mistake, respectively.  
We use the notation mid$(R)$ in subroutines $\nac$ and $\ad$ to denote 
the middle point of a real interval $R$, top$(R)$ to be the largest point in $R$, and bot$(R)$ to be the smallest value in $R$.  Finally, we assume the known objective function is $c=(0,0,1)$.  

\begin{algorithm}[h]
 \begin{algorithmic}
 \INPUT Any learning algorithm $\cL$ and bit precision $N$
 \OUTPUT Polytope $\cP$ that is consistent with $\cL$ making a mistake each day.
 \Procedure{$\adversary$}{$\cL,N$} 
\State Set $R_1^{(0)}=[0,1], R_2^{(0)}=[1,2]$. \Comment{Initialize}
	\For{$t = 1,\cdots, N$}		
		\State $\qquad \left(\left(\bp\utt,q\utt\right),r_1\utt,r_2\utt\right)\gets \nac(R_1\uttb , R_2\uttb)$.
		\State $\qquad$ Show constraint $\bp\utt\cdot\bbx \leq q\utt$ to $\cL$.  \Comment{Constraint} 
		\State $\qquad $ Get prediction $\hat\bbx\utt$ from $\cL$. 
		\State$\qquad \left(\bbx\utt, R_1\utt, R_2\utt\right)\gets\ad\left(R_1\uttb , R_2\uttb, r_1\utt,r_2\utt, \hat\bbx\utt\right)$. \Comment{Update}
		\State $\qquad $ Reveal the optimal $\bbx\utt\ne\hat\bbx\utt$ and update the regions $R_1\utt$ and $R_2\utt$.
	\EndFor
	\State $A,\bbb \gets \conmatr(R_1^N,R_2^N)$ \Comment{Constraint matrix consistent with $\{ \bbx\utt\mid t \in [N]\}$} \\
	\Return $A,\bbb$
 \EndProcedure
 \end{algorithmic}
 \caption{Adversary Updates ($\adversary$)}
 \label{alg:adversary}
\end{algorithm}

The procedure $\nac$ takes as input two real valued intervals and then outputs two points $r_1$ and $r_2$ as well as the new constraint denoted by the pair $(\bp, q)$. 
The two points will be used as input in $\ad$ along with the learner's prediction. 
In procedure $\ad$ the adversary makes sure that the learner suffers a mistake.
On each day, one of the points say $r_2$ produced by $\nac$ has a higher objective than the other one.
If the learner chooses $r_2$ then the adversary will simply choose a polytope that makes $r_2$ infeasible 
so that $r_1$ is actually the optimal point that day. If the learner chooses $r_1$ then the adversary picks $r_2$ as the optimal solution.
%The pair $(\bp,q)$ defines the new constraint $\bp\cdot\bbx\leq q$ that will be added on that day.  
Note that the three points $r_1$, $r_2$, and $r_3$ computed in $\nac$ all bind the constraint and are not collinear, and thus uniquely define the 
hyperplane $\{ \bbx : \bp\cdot \bbx = q\}$. Finally, in $\ad$ the adversary updates the new feasible region for her use in the next days.

%%%FIX
%\begin{algorithm}[h]
% \begin{algorithmic}
% \Procedure{$\nac$}{$R_1 , R_2$} 
%\State Set $\epsilon\leftarrow 0.01$.
%\State Set $r_1\gets \left(0,1, \text{mid}(R_1)\right)$
%\State $\qquad r_2 \gets \left(1, \text{mid}(R_2), 1+\epsilon\cdot\text{mid}(R_2)\right)$
%\State $\qquad  r_3\gets\left(1, \text{mid}(R_2), 0\right)$.
%\State Set $\bp = \left(1-\text{mid}(R_2),1,0 \right)$ and $q = 1$
%\State \hspace{49mm} $\triangleright$ Note that constraint will be $\bp \cdot \bbx \leq 1$ and binds at $r_1, r_2, r_3$
%\State \textbf{return} $(\bp,q)$ and $r_1, r_2$.
% \EndProcedure
% \end{algorithmic}
% \caption{New Adversarial Constraint ($\nac$)}
% \label{alg:adv0}
%\end{algorithm}

%%%FIX
%\begin{algorithm}[h]
%{\SetAlgoNoLine
% \begin{algorithmic}
% \Procedure{$\ad$}{$R_1 , R_2,r_1, r_2, \hat\bbx$} \SetAlgoNoEnd
%\State
%\eIf{ $\hat\bbx == r_2$ }{ 
%$\bbx\gets r_1$.
%$R_2 \leftarrow \left[\text{bot}(R_2), \text{mid}(R_2)\right]$. \Comment{$r_1$ and $r_2$ as in Algorithm~\ref{alg:adv0}}
%}{
%$\bbx\gets r_2$.
%$R_2 \leftarrow \left[\text{mid}(R_2), \text{top}(R_2)\right].$}
%\ifnum\colt=0
%\EndIf
%\fi
%\State{$R_1\leftarrow \left[\text{mid}(R_1), \text{top}(R_1)\right].$}
%\State{\textbf{return} $\bbx, R_1, R_2$}
% \EndProcedure
% \end{algorithmic}
% }
% \caption{Adaptive Adversary ($\ad$)}
% \label{alg:adv}
%\end{algorithm}
\begin{algorithm}[h]
 \begin{algorithmic}
 \Procedure{$\nac$}{$R_1 , R_2$} 
\State Set $\epsilon\leftarrow 0.01$.
\State Set $r_1\gets \left(0,1, \text{mid}(R_1)\right)$
\State $\qquad r_2 \gets \left(1, \text{mid}(R_2), 1+\epsilon\cdot\text{mid}(R_2)\right)$
\State $\qquad  r_3\gets\left(1, \text{mid}(R_2), 0\right)$.
\State Set $\bp = \left(1-\text{mid}(R_2),1,0 \right)$ and $q = 1$
\Comment{The constraint is $\bp \cdot \bbx \leq 1$ and binds at $r_1, r_2, r_3$}
\State \textbf{return} $(\bp,q)$ and $r_1, r_2$.
 \EndProcedure
 \end{algorithmic}
 \caption{New Adversarial Constraint ($\nac$)}
 \label{alg:adv0}
\end{algorithm}

\begin{algorithm}[h]
 \begin{algorithmic}
 \Procedure{$\ad$}{$R_1 , R_2,r_1, r_2, \hat\bbx$} 
\If{$\hat\bbx==r_2$}\Comment{$r_1$ and $r_2$ as in Algorithm~\ref{alg:adv0}}
\State $\bbx\gets r_1$.
\State $R_2 \leftarrow \left[\text{bot}(R_2), \text{mid}(R_2)\right].$
\Else
\State $\bbx=r_2$.
\State $R_2 \leftarrow \left[\text{mid}(R_2), \text{top}(R_2)\right].$
\EndIf
\State $R_1\leftarrow \left[\text{mid}(R_1), \text{top}(R_1)\right].$
\State \textbf{return} $\bbx, R_1, R_2$
 \EndProcedure
 \end{algorithmic}
 \caption{Adaptive Adversary ($\ad$)}
 \label{alg:adv}
\end{algorithm}

$\adversary$ finishes by actually outputting the polytope that was consistent with the constraints and the optimal solutions he showed at each day.  This
polytope is defined by constraint matrix $A$ and vector $\bbb$ using the subroutine $\conmatr$ as well as the nonnegativity constraint $\bbx \geq 0$.
%%%FIX
%\begin{algorithm}
%\begin{algorithmic}
%\Procedure{$\conmatr$}{$R_1,R_2$}
%\State Set $f_1 = \left(\text{top}(R_1)+\text{bot}(R_1)\right)/2$ and $f_2 = \left(\text{top}(R_2)+\text{bot}(R_2)\right)/2$ and $\epsilon>0$ 
%\State \[A\gets \left( \begin{array}{ccc}
%-1 & 0 & 0 \\
%1 & 0 & 0 \\
%(f_1-1-\epsilon) & -\epsilon & 1 \\
%-(f_2-1)\cdot f_1 & f_1 & 0
%\end{array} \right)
%\text{ and }
%\bbb \gets \left( \begin{array}{c}
%0 \\
%1 \\
%f_1-\epsilon\\
%f_1
%\end{array} \right).\]
%\Return $A$ and $\bbb$
%\EndProcedure
% \caption{Matrix consistent with adversaries reports ($\conmatr$)}
% \label{alg:matrix}
%\end{algorithmic}
%\end{algorithm}
\begin{algorithm}
\begin{algorithmic}
\Procedure{$\conmatr$}{$R_1,R_2$}
\State Set $f_1 = \left(\text{top}(R_1)+\text{bot}(R_1)\right)/2$ and $f_2 = \left(\text{top}(R_2)+\text{bot}(R_2)\right)/2$ and $\epsilon>0$ 
\State \[A\gets \left( \begin{array}{ccc}
-1 & 0 & 0 \\
1 & 0 & 0 \\
(f_1-1-\epsilon) & -\epsilon & 1 \\
-(f_2-1)\cdot f_1 & f_1 & 0
\end{array} \right)
\text{ and }
\bbb \gets \left( \begin{array}{c}
0 \\
1 \\
f_1-\epsilon\\
f_1
\end{array} \right).\]
\Return $A$ and $\bbb$
\EndProcedure
 \caption{Matrix consistent with adversary ($\conmatr$)}
 \label{alg:matrix}
\end{algorithmic}
\end{algorithm}

To prove that the procedure given in $\adversary$ does in fact make every learner $\cL$ make a mistake at every day, we need to show that 
\emph{(i)} there exists a simple unknown polytope that is consistent with what the adversary has presented in the previous days.
Furthermore, we need to show that \emph{(ii)} the optimal point returned by the adversary on each day is indeed the optimal point
corresponding to the LP with objective $\bbc$ and unknown constraints subject to the additional constraint added on each day.  

%Thus, for any number of days $N$ that the learner wishes to interact with $\adversary$, the learner will make a mistake at every day and the resulting polytope that causes this can be fully determined by matrix $A$ and vector $\bbb$.

To show \emph{(i)} note that point $r_1\utt= (0,1,\text{mid}(R_1\uttb))$ is always a feasible point for $t \in [N]$ in the polytope given by $A$ and $\bbb$ and the new constraint added each day will allow $r_1\utt$ to remain feasible. %Furthermore, the polytope $A\bbx\le\bbb$ returned by the adversary has only 5 edges (which could be as many as ${4\choose 2}$. 

To show \emph{(ii)} first note that 
the new constraint added is always a binding constraint. So
by Assumption~\ref{assumpt:bind}, it is sufficient to check the intersection of the edges of the polytope output by $\conmatr$ and the newly added hyperplane and return the (feasible) point with the highest objective as the optimal point.
Second, the following equations define the edges of the polytope which are one dimensional subspaces $\bbe^{i,j}$ according to Assumption \ref{assumpt:rank} with $A$ and $\bbb$ being the output of $\conmatr$.
$$
\bbe^{i,j} = \{\bbx\in \R^3 \mid A_{i}\bbx= b_{i}\quad \text{ and } \quad A_j \bbx = b_j \} \qquad i,j \in \{1,2,3 ,4\}, \quad i \neq j,
$$
where $A_i$ is the $i$th row of $A$.  Since the first two constraints define two parallel hyperplanes, we only need to consider 5 edges. Let
$$r_1\utt = \left(0,1,\text{mid}(R_1\uttb )\right),$$ and $$r_2\utt = \left(1,\text{mid}(R_2\uttb),1+\epsilon\cdot \text{mid}(R_2\uttb)\right).$$ We show 
that the new constraint either intersects the edges of the polytope at $r_1\utt$ or $r_2\utt$ or do not intersect with them at all. This will prove that 
the optimal points shown by the adversary each day is consistent with the unknown polytope.

\begin{enumerate}
\item $\bbe^{1,4}=  \left(0,0,f_1\right)\cdot s+\left(0, 1, 0\right)$ that intersects with the new hyperplane at $r_1\utt$. %at $t=\text{mid}(R_1)/f_1$.
\item $\bbe^{2,3} = \left(0, f_2, \epsilon\cdot f_2\right)\cdot s + \left(1, 0, 1\right)$ that intersects with the new hyperplane at $r_2\utt$. %(when $t=\text{mid}(R_2)/f_2$).
\item $\bbe^{2,4}=  \left(0,0,1+\epsilon\cdot f_2\right)\cdot s+\left(1, f_2, 0\right)$ that does not intersect the new hyperplane unless $\text{mid}(R_2)=f_2$ (which does not happen).
\item $\bbe^{3,4} = \left(1, f_2-1, 1+\epsilon\cdot f_2-f_1\right)\cdot s+\left(0, 1, f_1\right)$ that does not intersect the new hyperplane unless $\text{mid}(R_2)=f_2$ (which does not happen).
\item $\bbe^{1,3}=  \left(0, -1, f_1\right)\cdot s + \left(0, 1, f_1\right) $ never intersects the hyperplane.
\end{enumerate}
And this concludes the proof.
\end{proof}

We now prove Theorem \ref{thm:impossibility} even for $d>3$.  \\
\textbf{Proof of Theorem \ref{thm:impossibility}.}
We modify the proof of Theorem~\ref{thm:imp-3d} to $d>3$ by adding dummy variables. These dummy variables are denoted by $x_{4:d}$.
Furthermore, we add dummy  constraints $x_i\ge 0$ for all the dummy variables. 
We modify the objective function in the proof of Theorem~\ref{thm:imp-3d} to be $\bbc = (0, 0, 1, -1, \ldots, -1)$. This will cause all the newly added
variables to have no effect on the optimization (they should be set to $0$ in the optimal solution) and, hence, the result from Theorem~\ref{thm:imp-3d}
extends to the case when $d>3$.
\qed
%%%%%%%%%%%%%%%%%%%%%%%%%%%%%%%%%
%%%%% SToCHASTIC SECTION%%%%%%%%%%%%%%%%
%%%%%%%%%%%%%%%%%%%%%%%%%%%%%%%%%
\subsection{Section~\ref{sec:stochastic}}

\textbf{Proof of Lemma~\ref{lem:stoch_case}.}
%Let $E$ denote the set of edges of $P$, and
%\ar{Changed this proof substantially}\sj{thanks!}
First, since all of the points $x_1,\ldots,x_t$ are drawn i.i.d. from $\cD$, we observe by symmetry that the event we are interested in is distributed identically to the following event: draw a set of $t$ points $X'=\{x_1,\ldots, x_{t}\}$ i.i.d. from $\cD$ and select an index $i \in \{1,\ldots,t\}$ uniformly at random and compute the probability that $x_i \not\in \mathrm{Conv}(X' \setminus \{x_i\})$. In other words
 \begin{equation}
 \label{eq:symm}
 \Pr_{x_1,\ldots,x_t \sim \cD}[x_t\not\in \mathrm{Conv}(X)] = \Pr_{x_1,\ldots,x_t \sim \cD, i \sim \{1,\ldots,t\}}\left[x_i\not\in \mathrm{Conv}(X' \setminus \{x_i\})\right].
 \end{equation}
We analyze the quantity on the right hand side of~(\ref{eq:symm}) instead, fixing the choices of $x_1,\ldots,x_t$, and analyzing the probability only over the randomness of the choice of index $i$.
For each edge $e\in E_{\cP}$, let $X_e' = X' \cap e$.  %be the set of points on $e$ which are in $X'$. 
Since each edge lies on a one dimensional subspace, there are at most two extreme points $x^e_1,x^e_2 \in X_e'$
that lie outside of the convex hull of other points i.e. such that $x^e_1 \not\in \mathrm{Conv}(X'\setminus\{x^e_1\})$ and $x^e_2 \not\in \mathrm{Conv}(X'\setminus\{x^e_2\})$.
%% A point is an extreme point among a set of points if it is not in
%% the convex hull of other points, and there are at most two extreme
%% points
We note that when we choose an index $i$ uniformly at random, the probability that we select a point $x \in X'_e$ is exactly $|X'_e|/t$, and conditioned on selecting a point $x \in X'_e$, the probability that $x$ is an extreme point (i.e.  $x \in \{x^e_1,x^e_2\}$) is at most $2/|X'_e|$.
Hence, we can calculate
\begin{align*}
\Pr\left[ x_i \not \in \mathrm{Conv}(X'\setminus \{x_i\}) \right] &=
\sum_{e\in E_{\cP}} \Pr[x_i \in X'_e] \cdot \Pr[ x_i \notin
  \mathrm{Conv}(X_e'\setminus \{x_i\}) \mid x_i \in X'_e]\\
&\leq \sum_{e\in E_{\cP}} \frac{|X_e'|}{t} \cdot \frac{2}{|X_e'|} = \sum_{e\in E_{\cP}} \frac{2}{t} = \frac{2|E_{\cP}|}{t}.
\end{align*} 
\qed

\textbf{Proof of Theorem~\ref{thm:mistakebdd}.}
First, we show that $\learnhull$ makes a mistake only if the true
optimal point $\bbx^{(t)}$ lies outside of the convex hull
$\cC^{(t-1)}$ formed by the previous observed optimal points $\{\bbx^{(1)}, \ldots, \bbx^{(t-1)}\}$. Suppose that at
day $t$, the algorithm predicts the point $\hat \bbx^{(t)}$ instead
of the optimal point $\bbx^{(t)}$. Since each point in $\cC^{(t-1)}$
is feasible and $\hat \bbx^{(t)}$ is the point with the highest
objective value among the points in
$\{\bbx \in \cC^{(t-1)} \mid \mathbf{p}' \cdot \bbx\leq b'\}$, then it
must be that $\bbx^{(t)}\not\in \cC^{(t-1)}$ because otherwise $\bbc\cdot\bbx^{(t)} > \bbc\cdot\hat\bbx^{(t)}$. By
Lemma~\ref{lem:stoch_case}, we also know that the probability that
$\bbx^{(t)}$ lies outside of $\cC^{(t-1)}$ is no more than $2|E_{\cP}|/t$ in expectation,
which also upper bounds the probability of $\learnhull$ making a
mistake at day $t$. Therefore, the expected number of mistakes made
by $\learnhull$ over $T$ days is bounded by the sum of probabilities of making a mistake in each day which is $\sum_{t=1}^T 2|E_{\cP}|/t =
O(|E_{\cP}|\log(T))$.

\qed
%\textbf{Proof of Lemma~\ref{lem:stoch_case}}
%\input{./lem_stoch_case}
%\qed

\textbf{Proof of Theorem~\ref{thm:stoch}.}
The deterministic procedure runs $\lceil 18\log(1/\delta)\rceil$ independent instances of the $\learnhull$ each using independently drawn examples. The independent instances are aggregated into a single prediction rule by predicting using the \emph{modal} \rynote{modal?} prediction (if one exists), and otherwise predicting arbitrarily. Hence, the aggregate prediction is correct whenever at least half of the instances of $\learnhull$ are correct.

We show that if each instance of the $\learnhull$ is run for $8|E_{\cP}|$ days, then the probability that more than half of the instances of $\learnhull$ make a mistake on a newly drawn constraint at day $T+1$ is at most $\delta$. The result is that with probability at least $1-\delta$
the majority of instances of $\learnhull$ predict the correct optimal point, and hence the aggregate prediction is also correct.

Let $Z_i$ be the random variable that denotes the probability that the $i$th instance of the $\learnhull$ algorithm makes a mistake on a fresh example, after it has been trained \rynote{Trained?  Is this formal?}
for $8|E_{\cP}|$ days.
By Theorem~\ref{thm:mistakebdd}, we know $E[Z_i]\leq 1/4$ for all $i$.
%$$Z = \frac{1}{|X\utt|} \cdot \sum_{x \in X\utt} \1\{ \bbx \not \in \text{conv}(X\utt \backslash \bbx) \}.$$
Now by Markov's inequality,
% and using the result of Lemma \ref{lem:stoch_case} we have
$$
\Pr\left[Z_i \geq \frac{3}{4}\right] = \Pr\Big[Z_i \geq 3\cdot E[Z_i]\Big] \leq 1/3,
$$
for all $i$. Hence, the \emph{expected} number of instances that make a mistake is at most $1/3$. Finally, since each instance is trained on independent examples, a Chernoff bound implies that the probability that at least half of the instances of $\learnhull$ make a mistake is bounded by $\delta$.

%\ar{Do we need to invoke some mysterious theorem here? Doesn't this just follow from a Chernoff bound?}\sj{reworked}
%We then run $T= \log(1/\delta)$ independent trials of $\learnhull$ simultaneously each with $t$ fresh draws from the distribution $\cD$. We denote each trial $i$ with its own proportion $Z_i$ for its sample of points $X\utt_i$.  We use a Chernoff bound to then bound the number of trials where $Z_i \geq 2\beta$ to get
%$$
%\Prob{}{\frac{1}{T} |\{i \in [T] : Z_i \geq 2\beta \} |\geq 1/2+\epsilon} .
%$$
%Then the probability that in at least in half of the copies of the $\learnhull$, the predicted optimal point is incorrect is new optimal point is outside
%of the convex hull of at least half of the trials of $\learnhull$ is simply bounded by $(1/2)^{(\log(1/\delta))}=\delta$ by independence.
%% because the trials of $\learnhull$ use independent samples.
%
%In Theorem~\ref{thm:mistakebdd}, we showed that the $\learnhull$ only makes a mistake if the optimal point is outside of its convex hull.
%So with probability at least $1-\delta$, at least half of the trials of the $\learnhull$ predict the optimal point correctly.

%Finally, the sample complexity
%of running $6\log(1/\delta)/\beta$ trials of the $\learnhull$ is $\lceil 12|E|/\beta^2\cdot\log\left(1/\delta\right)\rceil$ as in the statement of the theorem.  Our procedure then is to run $T$ independent trials of $\learnhull$ where each trial uses $t$ samples.  
\qed
\section{Circumventing the Lower Bound when $d\leq 2$}
\label{sec:1-2d}
In Theorem~\ref{thm:lower} (in Section~\ref{sec:lowerbound}), 
we proved the necessity of
Assumption~\ref{assumpt:precision} by showing that the dependence on
the precision parameter $N$ in our mistake bound is tight. However, 
Theorem~\ref{thm:lower} requires the dimension $d$ to be at least 3.

We now show that this condition on the dimension is indeed
necessary---even without the finite precision assumption
(Assumption~\ref{assumpt:precision}), we can have (computationally
efficient) algorithms with small mistake bounds when the dimension
$d\leq 2$.

In $d=1$, at most two constraints are sufficient to determine any constraint matrix $A$ because 
the constraint matrix $A$ defines a feasible interval on the real line. So we will guess the value that maximizes the objective subject to the single known constraint.  Once we have made a mistake, we must have learned the true optimal to the underlying problem because our guess was infeasible.  After this single mistake, we either guess the true optimal that we have already seen or if it is not feasible with the new constraint then we guess the point that maximizes the objective subject to the new constraint.  Thus after one mistake, the learner will not make any more mistakes. 

Lemma \ref{lem:collinear} tells us that the line between any three collinear points must give us an edge-space of the underlying polytope.  When $d=2$, the corresponding edge-space is then just one of the original constraints of the underlying polytope.  Since each solution must be on an edge of the underlying polytope each day, we can make at most $3m$ mistakes without seeing the true objective.  Hence, all together, we can make at most $3m+1$ mistakes before we recover all the constraints of the underlying polytope, or all the rows of the constraint matrix $A$, and see the true optimal solution.  

This phenomenon does not continue to hold for $d > 2$ (as we show in our lower bound in Theorem~\ref{thm:lower}).
%%% Local Variables:
%%% mode: latex
%%% TeX-master: t
%%% End:

\section{Missing Proofs from Section~\ref{sec:known_constr}}
\label{sec:missingproof4}
\ifnum\final=0
First we state Theorem~\ref{thm:ell_bits} from~Grotschel et al.~\cite{GLS93} about the running time of the Ellipsoid algorithm.
\begin{theorem}%[Grotschel et al.~\cite{GLS93}] 
Let $\cP \subset \R^d$ be a polytope given as the intersection of linear constraints, each specified with $N$ bits of precision.  Given access to a separation oracle which can return, for each candidate solution $p \notin \cP$ a hyperplane with $N$ bits of precision  that separates $p$ from $\cP$, the Ellipsoid algorithm outputs a point $p' \in \cP$ or outputs $\cP$ is empty at most $\poly(d, N)$ iterations.
\label{thm:ell_bits}
\end{theorem}
We are now ready to bound the number of mistakes that $\learnell$ makes. \\
\ifnum\final=0
\textbf{Proof of Theorem~\ref{thm:unknwon-obj}.}
Whenever $\learnell$ makes a mistake, there exists a separating hyperplane 
\begin{equation*}
\cH\utt = \left\{W = (\bbw^1,\ldots, \bbw^n) \in \R^{n\times d} \mid \sum_{i \in S\utt}\bbw^i \cdot (\bbx\utt - \hat\bbx\utt ) > 0 \right\}
\label{eq:sep_hyp*}
\end{equation*}
 that cause the Ellipsoid algorithm to run for another iteration.  
When $\learnell$ reduces $\cF$ to a set that contains a single point up to $N$ bits of precision then predicting via the above equation for $\cH\utt$ will ensure we never make a mistake again.  Hence, the number of mistakes that $\learnell$ can commit against an adversary is bounded by the maximum number of iterations for which the Ellipsoid algorithm can be made to run, in the worst case. 

We know that $\bbx\utt$ each day is on a vertex of the polytope $\cP\utt$ which is guaranteed 
to have coordinates specified with at most $N$ bits of precision by Assumption~\ref{assumpt:precision2}.   
Theorem \ref{thm:lin_solve} guarantees that the solution $\hat\bbx\utt$ that $\learnell$ produces using the following equation 
\begin{equation*}
\hat\bbx\utt \in \myargmax_{\bbx \in \cP\utt} \left\{ \sum_{i \in S\utt}  (\bbw^i)\utt \cdot \bbx\right\}
\label{eq:obj_guess2}
\end{equation*}
is a vertex solution of $\cP\utt$. So $\hat\bbx\utt$
can also be written with $N$ bits of precision by Assumption~\ref{assumpt:precision2}.  Thus, every constraint in $\cF$ and hence each separating hyperplane can be written with $d\cdot  N$ bits of precision.  By Assumption \ref{assumpt:obj_precision} and Theorem \ref{thm:ell_bits}, we know that the Ellipsoid algorithm will find a point in the feasible region $\cF$ after at most $\poly(n,d,N)$ many iterations which is the same as the number of mistakes of $\learnell$.
\qed

\end{document}